\documentclass[11pt]{article}%
\usepackage{amsfonts}
\usepackage{amsmath}
\usepackage{amssymb}
\usepackage{graphicx}
\usepackage{xcolor}%
\providecommand{\U}[1]{\protect\rule{.1in}{.1in}}
\newtheorem{theorem}{Theorem}
\newtheorem{acknowledgement}[theorem]{Acknowledgement}

\newtheorem{corollary}[theorem]{Corollary}

\newtheorem{lemma}[theorem]{Lemma}
\newtheorem{notation}[theorem]{Notation}

\newtheorem{proposition}[theorem]{Proposition}
\newtheorem{remark}[theorem]{Remark}

\newenvironment{proof}[1][Proof]{\noindent\textbf{#1.} }{\ \rule{0.5em}{0.5em}}
\ifx\pdfoutput\relax\let\pdfoutput=\undefined\fi
\newcount\msipdfoutput
\ifx\pdfoutput\undefined\else
\ifcase\pdfoutput\else
\ifx\paperwidth\undefined\else
\ifdim\paperheight=0pt\relax\else\pdfpageheight\paperheight\fi
\ifdim\paperwidth=0pt\relax\else\pdfpagewidth\paperwidth\fi
\fi\fi\fi
\begin{document}

\title{Relative Phase Shifts for Metaplectic Isotopies Acting on Mixed Gaussian States}
\author{Maurice A. de Gosson\thanks{maurice.de.gosson@univie.ac.at}\\University of Vienna\\Faculty of Mathematics (NuHAG)
\and Fernando Nicacio\thanks{nicacio@if.ufrj.br}\\Universidade Federal de Rio de Janeiro\\Instituto de F\'{\i}sica}
\maketitle

\begin{abstract}
We address in this paper the notion of relative phase shift for mixed quantum
systems. We study the Pancharatnam--Sj\"{o}qvist phase shift $\varphi
(t)=\operatorname*{Arg}\operatorname*{Tr}(\widehat{U}_{t}\widehat{\rho})$ for
metaplectic isotopies acting on Gaussian mixed states. We complete and
generalize previous results obtained by one of us while giving rigorous
proofs. This gives us the opportunity to review and complement the theory of
the Conley--Zehnder index which plays an essential role in the determination
of phase shifts.

\end{abstract}

\section{Introduction}

While the postulates of quantum mechanics seem to recognize complex wave
functions as mere instruments for calculating probability amplitudes, their
phases should definitively not be viewed as secondary objects. There is
actually a plethora of examples in which the phase plays the title-role. The
arguably most famous example of this is the Aharonov--Bohm effect dealing with
questions about the factual significance of electromagnetic potentials
\cite{aharonov-bohm}. A break point on general phases behavior in quantum
mechanics also emerges from Berry's seminal work \cite{berry}. Its main
contribution is to recognize that the total phase of a system is composed by
two essentially distinct effects: a phase related to the system dynamics, and
a geometrical phase, which mirrors the geometry of the underlying Hilbert
space of the system of pure states. The same kind of geometrical phenomenon
was observed earlier by Pancharatnam \cite{pancha} in the context of classical
optics, where the phase shift is due to the spherical geometry of the
polarization states of light. Also, the Hannay angle \cite{hannay} is an
example of phase shifts induced by the space shape in classical mechanics.

A more profound and general knowledge was acquired by Mukunda and Simon
\cite{mukunda}, where the authors gave precise definitions for the total, the
dynamical, and the geometrical phases for pure states only as functions of
paths in the Hilbert space.

Consider a quantum system represented at initial time $t=0$ by a function
$\psi\in L^{2}(\mathbb{R}^{n})$. Assuming that the time-evolution of the
system is governed by a one-parameter family of unitary operators
$\widehat{U}_{t}$ on $L^{2}(\mathbb{R}^{n})$ the system will be represented at
time $t$ by the function $\psi_{t}=\widehat{U}_{t}\psi$.{ In \cite{mukunda}
Mukunda and Simon (see also \cite{aharonov}) }defined the (relative) phase
shift of the system when $(\widehat{U}_{t}\psi|\psi)_{L^{2}}\neq0$ by the
formula%
\begin{equation}
\varphi(t)=\operatorname*{Arg}(\widehat{U}_{t}\psi|\psi)_{L^{2}},
\label{shift1}%
\end{equation}
{which was named as the Pancharatnam or total phase. }

Suppose now that the system under consideration is in a \textquotedblleft
mixed state\textquotedblright\ represented by a density operator
$\widehat{\rho}=\sum_{j}\lambda_{j}\widehat{\Pi}_{\psi_{j}}$ ($\widehat{\Pi
}_{\psi_{j}}$ the orthogonal projection on the ray $\mathbb{C}\psi_{j}$). The
operator $\widehat{\rho}$ is a positive semidefinite (and hence self-adjoint)
trace class operator on\ $L^{2}(\mathbb{R}^{n})$ with trace
$\operatorname*{Tr}(\widehat{\rho})=1$. Its time evolution is given by
$\widehat{\rho}_{t}=\widehat{U}_{t}\widehat{\rho}\widehat{U}_{t}^{\ast}$ and
one now defines, following Sj\"{o}qvist \textit{et} \textit{al.}
\cite{sjoeqvist}, the phase shift of this quantum state by
\begin{equation}
\varphi(t)=\operatorname*{Arg}\operatorname*{Tr}(\widehat{U}_{t}\widehat{\rho
}) \label{shift2}%
\end{equation}
when $\operatorname*{Tr}(\widehat{U}_{t}\widehat{\rho})\neq0$. It is easy to
see that this definition coincides with Pancharatnam's formula (\ref{shift1})
when $\widehat{\rho}=\widehat{\Pi}_{\psi}$: since trace class operators form a
two-sided ideal in the algebra of bounded operators the product $\widehat{U}%
_{t}\widehat{\rho}$ is a trace class operator with same rank one as
$\widehat{\rho}$; it follows that $\operatorname*{Tr}(\widehat{U}%
_{t}\widehat{\rho})$ is precisely the only eigenvalue of this operator. The
equation $\widehat{U}_{t}\widehat{\rho}\phi=\lambda\phi$ is equivalent to
$(\phi|\psi)_{L^{2}}\widehat{U}_{t}\psi=\lambda\phi$. Choosing $\phi
=\widehat{U}_{t}\psi$ we get $(\widehat{U}_{t}\psi|\psi)_{L^{2}}%
\widehat{U}_{t}\psi=\lambda\widehat{U}_{t}\psi$ hence $\lambda=(\widehat{U}%
_{t}\psi|\psi)_{L^{2}}$. {The generalization (\ref{shift2}) relies on the fact
that this quantity, as well as the one in (\ref{shift1}), can be defined and
measured by interferometric techniques as explained in
\cite{sjoeqvist,nicacio1}.}

We will study in this paper the Pancharatnam--Sj\"{o}qvist phase shift when
the Wigner distribution of $\widehat{\rho}$ is of the type
\begin{equation}
\rho(z)=(2\pi)^{-n}\sqrt{\det V^{-1}}e^{-\frac{1}{2}V^{-1}z\cdot z}
\label{rhosig1}%
\end{equation}
and $(\widehat{U}_{t})$ is the Schr\"{o}dinger evolution operator determined
by a time-de\-pen\-dent quadratic Hamiltonian. Gaussians distributions of the
type (\ref{rhosig1}) play a central role in quantum mechanics and optics, and
are paradigmatic for all other states. We thus extend the results obtained by
one of us in the recent work \cite{nicacio1}; this allows us in particular to
give precise formulas for the harmonic oscillator in $n$ dimensions.

This work is structured as follows:

\begin{itemize}
\item In section \ref{sec1} we introduce the notion of symplectic isotopy: a
symplectic isotopy is a $C^{1}$-path of symplectic matrices passing through
the origin at time $t=0.$ This notion generalizes that of one-parameter group;
we show that a symplectic isotopy can always be viewed a the Hamiltonian flow
of a (possibly time-dependent) Hamiltonian that is a quadratic form in the
position and momentum variables. To every symplectic isotopy is associated in
a canonical way a $C^{1}$-path of metaplectic operators; this allows the
derivation of Schr\"{o}dinger's equation for quadratic Hamiltonians
\cite{Folland,Birk};

\item In section \ref{secmeta} we review the properties of the Weyl symbol of
metaplectic operators as developed by one of us \cite{LMP,Birk,RMPCZ}; these
properties will be instrumental for our derivation of phase formulas. This
gives us the opportunity to present \textquotedblleft in a
nutshell\textquotedblright\ a rather technical topic which is not very
well-known outside mathematicians working on symplectic geometry and
intersection theory; this section begins by a review of the general notion of
Weyl transform;

\item In section \ref{secdensity} we review the basic properties of density
operators we will need, focusing in particular on the Gaussian case, which is
of great practical interest in quantum optics. Among all quantum states,
Gaussian states are those whose properties are the best understood from a
theoretical point of view; they play a significant role in many areas of
quantum mechanics and optics, quantum chemistry, and signal theory.

\item Section \ref{secpancha} is devoted to the study of the
Pancharatnam--Sj\"{o}qvist phase shift when the Hamiltonian flow is determined
by a quadratic Hamiltonian and acts on a density operator with Gaussian Wigner
distribution. We prove a general formula for the action of metaplectic
operators on Gaussian density matrix and thereafter give detailed calculations
for the harmonic oscillator.

\item In section \ref{tic} we generalize the previous results to the
Inhomogeneous metaplectic group, taking into account affine transformations
related to displacements in phase space.
\end{itemize}

To make the paper self-contained we have carefully detailed the construction
and the properties of the Conley--Zehnder index, and added two Appendices: in
Appendix A we collect the main definitions and properties of the metaplectic
group, and in Appendix B we review the theory of the Leray--Maslov index which
plays an essential role in the definition of the extended Conley--Zehnder
intersection index of symplectic paths without restrictions on the endpoint of
these paths.

\paragraph{Notation and prerequisites}

The standard symplectic form on $\mathbb{R}^{2n}$ is $\sigma=\sum_{j=1}%
^{n}dp_{j}\wedge dx_{j}$, that is $\sigma(z,z^{\prime})=p\cdot x^{\prime
}-p^{\prime}\cdot x$ if $z=(x,p)$, $z^{\prime}=(x^{\prime},p^{\prime})$; in
vector notation $\sigma(z,z^{\prime})=Jz\cdot z^{\prime}=(z^{\prime})^{T}Jz$
where $J=%
\begin{pmatrix}
0_{n\times n} & I_{n\times n}\\
-I_{n\times n} & 0_{n\times n}%
\end{pmatrix}
$. The scalar product on the space $L^{2}(\mathbb{R}^{n})$ is defined by%
\[
(\psi|\phi)_{L^{2}}=\int\psi(x)\overline{\phi(x)}d^{n}x.
\]

Let $Q$ be a real quadratic form on $\mathbb{R}^{m}$. The signature
$\operatorname*{sign}(Q)$ is the number of $>0$ eigenvalues of the Hessian
matrix of $Q$ minus the number of $<0$ eigenvalues. We will use the
generalized Fresnel formula (see Appendix A of Folland \cite{Folland})
\begin{equation}
\int e^{-\frac{1}{2\hbar}Az\cdot z}d^{2n}z=(2\pi\hbar)^{n}\det\nolimits^{-1/2}%
A \label{Fresnel}%
\end{equation}
which is valid for all $A=A^{\ast}$ with $\operatorname{Re}A>0$ and where
$(\det A)^{-1/2}=\alpha_{1}^{-1/2}\cdot\cdot\cdot\alpha_{2n}^{-1/2}$, the
$\alpha_{j}^{-1/2}$ being the square roots of $\alpha_{j}^{-1}$ with positive
real part.

\section{\label{sec1}Symplectic and Metaplectic Isotopies}

\subsection{Hamiltonian and symplectic isotopies}

A symplectomorphism of $\mathbb{R}^{2n}$ is a $C^{\infty}$ diffeomorphism
$f:\mathbb{R}^{2n}\longrightarrow\mathbb{R}^{2n}$ such that $f^{\ast}%
\sigma=\sigma$; equivalently the Jacobian matrix $Df(z)\in\operatorname*{Sp}%
(n)$ for every $z\in\mathbb{R}^{2n}$. If in addition there exists a
Hamiltonian function $H\in C^{\infty}(\mathbb{R}^{2n}\times\mathbb{R}%
,\mathbb{R})$ such that $f=f_{1}^{H}$ ($(f_{t}^{H})$ the flow determined by
the Hamilton equations for $H$) then $f$ is called a \emph{Hamiltonian
symplectomorphism}. A symplectic isotopy is a one-parameter family
$V=(S_{t})_{t\in I}$ of elements of $\operatorname*{Sp}(n)$ depending in a
$C^{1}$ fashion on $t\in I$ where $I$ is some real interval containing $0$ and
such that $f_{0}=I_{\mathrm{d}}$. The interval $I$ can be bounded, or
unbounded. If each $f_{t}$ is a Hamiltonian symplectomorphism, then
$\Sigma=(f_{t})_{t\in I}$ is called a Hamiltonian isotopy.

It is immediate to check that if $S_{t}=e^{tX}$ with $X\in\mathfrak{sp}(n)$
(the symplectic Lie algebra) then $(S_{t})_{t\in\mathbb{R}}$ is a genuine
one-parameter subgroup of $\operatorname*{Sp}(n)$, in fact the flow determined
by the quadratic Hamiltonian $H=-\frac{1}{2}JXz\cdot z$. It turns out that
each symplectic isotopy is a Hamiltonian isotopy determined by some
time-dependent $H$ (we are following here the presentation in \cite{RMP}):

\begin{proposition}
\label{prop5}Let $(f_{t})_{t\in I}$ be a Hamiltonian isotopy. We have
$(f_{t})_{t\in I}=(f_{t}^{H})_{t\in I}$ with%
\begin{equation}
H(z,t)=-\int_{0}^{1}\sigma(\dot{f}_{t}\circ f_{t}^{-1}(\lambda z),z)d\lambda
\label{hzt}%
\end{equation}
where $\dot{f}_{t}=df_{t}/dt$. Equivalently:%
\begin{equation}
H(z,t)=-\int_{0}^{1}\sigma\left(  X_{H}(f_{t}^{-1}(\lambda z),z\right)
)d\lambda\label{hztbis}%
\end{equation}
where $X_{H}=J\partial_{z}H$ is the (time-dependent) Hamilton vector field of
$H$.
\end{proposition}

\begin{proof}
See Wang \cite{Wang}; on a more conceptual level see Banyaga \cite{Banyaga}.
\end{proof}

In the case of general linear symplectic isotopies we have:

\begin{corollary}
\label{cor2}Let $\Sigma=(S_{t})_{t\in\mathbb{R}}$ be a symplectic isotopy in
$\operatorname*{Sp}(n)$.

(i) The associated Hamiltonian function is the quadratic form
\begin{equation}
H(z,t)=-\frac{1}{2}J\dot{S}_{t}S_{t}^{-1}z\cdot z=\frac{1}{2}\sigma(z,J\dot
{S}_{t}S_{t}^{-1}z) \label{hamzo}%
\end{equation}
where $\dot{S}_{t}=dS_{t}/dt$.

(ii) Writing $S_{t}$ in block-matrix form
\begin{equation}
S_{t}=%
\begin{pmatrix}
A_{t} & B_{t}\\
C_{t} & D_{t}%
\end{pmatrix}
\label{abcd}%
\end{equation}
that Hamiltonian is explicitly given by%
\begin{equation}
H=\tfrac{1}{2}(\dot{D}_{t}C_{t}^{T}-\dot{C}_{t}D_{t}^{T})x^{2}+(\dot{C}%
_{t}B_{t}^{T}-\dot{D}_{t}A_{t}^{T})p\cdot x+\tfrac{1}{2}(\dot{B}_{t}A_{t}%
^{T}-\dot{A}_{t}B_{t}^{T})p^{2}. \label{hamzi}%
\end{equation}

\end{corollary}

\begin{proof}
(i) Applying formula (\ref{hzt}) we get%
\begin{equation}
H(z,t)=-\int_{0}^{1}\sigma\left(  \dot{S}_{t}S_{t}^{-1}(\lambda z),z\right)
d\lambda\label{hamzobis}%
\end{equation}
which yields%
\[
H(z,t)=\frac{1}{2}\sigma(z,J\dot{S}_{t}S_{t}^{-1}z)
\]
hence (\ref{hamzo}), taking into account the linearity of $\sigma$ and $S_{t}%
$. (ii) It follows from the identity $S_{t}JS_{t}^{T}=J$ that%
\begin{equation}
S_{t}^{-1}=%
\begin{pmatrix}
D_{t}^{T} & -B_{t}^{T}\\
-C_{t}^{T} & A_{t}^{T}%
\end{pmatrix}
\end{equation}
and hence%
\[
J\dot{S}_{t}S_{t}^{-1}=%
\begin{pmatrix}
\dot{C}_{t}D_{t}^{T}-\dot{D}_{t}C_{t}^{T} & \dot{D}_{t}A_{t}^{T}-\dot{C}%
_{t}B_{t}^{T}\\
\dot{B}_{t}C_{t}^{T}-\dot{A}_{t}D_{t}^{T} & \dot{A}_{t}B_{t}^{T}-\dot{B}%
_{t}A_{t}^{T}%
\end{pmatrix}
;
\]
formula (\ref{hamzi}) readily follows.
\end{proof}

\subsection{The Conley--Zehnder index of a symplectic isotopy}

The Conley--Zehnder index $i_{\mathrm{CZ}}(\Sigma)$ for symplectic paths was
introduced in \cite{CZ} in the context of the study of Hamiltonian periodic
orbits in $\mathbb{R}^{2n}$. Meinrenken \cite{mein} has considerably extended
this index and applied it to Gutzwiller-type trace formulas \cite{mein1} (also
see the recent papers \cite{deng,shasun}). In \cite{useful} one of us has
shown that Conley--Zehnder index can be viewed as a particular case of an
index due to\ Leray which generalizes the Maslov index.

The vocation of $i_{\mathrm{CZ}}(\Sigma)$ is to count the intersections of a
symplectic path $\Sigma=(S_{t})_{0\leq t\leq1}$ with the manifold
\[
\operatorname*{Sp}\nolimits_{0}(n)=\{S\in\operatorname*{Sp}(n):\det(S-I)=0\}.
\]
Let us describe the original construction of the Conley--Zehnder index; we
will need for that some preparatory material. Consider the following subsets
of $\operatorname*{Sp}(n)$:
\begin{align*}
\operatorname*{Sp}\nolimits_{+}(n)  &  =\{S\in\operatorname*{Sp}%
(n):\det(S-I)>0\}\\
\operatorname*{Sp}\nolimits_{-}(n)  &  =\{S\in\operatorname*{Sp}%
(n):\det(S-I)<0\};
\end{align*}
these sets are connected and contractible and, together with
$\operatorname*{Sp}\nolimits_{0}(n)$ form a partition of the symplectic group:%
\[
\operatorname*{Sp}(n)=\operatorname*{Sp}\nolimits_{0}(n)\cup\operatorname*{Sp}%
\nolimits_{+}(n)\cup\operatorname*{Sp}\nolimits_{-}(n).
\]
Consider now the particular symplectic matrices $S_{+}$ and $S_{-}$ defined
by
\[
S_{+}=-I\text{ \ and \ }S_{-}=%
\begin{pmatrix}
L & 0\\
0 & L^{-1}%
\end{pmatrix}
\text{ , \ }L=\operatorname*{diag}(2,-1,...,-1);
\]
it is straightforward to check that we have $S_{+}\in\operatorname*{Sp}%
\nolimits_{+}(n)$ and $S_{-}\in\operatorname*{Sp}\nolimits_{-}(n)$. Define now
an extension $\widetilde{\Sigma}$ of the symplectic path $\Sigma$ by%
\[
\widetilde{\Sigma}(t)=\left\{
\begin{array}
[c]{c}%
S_{t}\text{ , \ }0\leq t\leq1\\
S_{t}^{\prime}\text{ , \ }1\leq t\leq2
\end{array}
\right.
\]
where the $S_{t}^{\prime}\in\operatorname*{Sp}(n)$ are defined as follows: if
$S\in\operatorname*{Sp}\nolimits_{+}(n)$ then $(S_{t}^{\prime})_{1\leq t\leq
2}$ is a continuous path joining $S=S_{1}$ to $S_{+}$ in $\operatorname*{Sp}%
\nolimits_{+}(n)$, and if $S\in\operatorname*{Sp}\nolimits_{-}(n)$ then
$(S_{t}^{\prime})_{1\leq t\leq2}$ is a path joining $S$ to $S_{-}$in
$\operatorname*{Sp}\nolimits_{-}(n)$. Recalling \cite{Folland,Birk} that every
$S\in\operatorname*{Sp}(n)$ has a polar decomposition $S=PR$ where
$P=(S^{T}S)^{1/2}\in\operatorname*{Sp}(n)$ is positive definite and
$R=(S^{T}S)^{-1/2}S$ is in the unitary subgroup $U(n)$ of $\operatorname*{Sp}%
(n)$ \textit{(i.e.} $R^{T}R=RR^{T}=I$, see \cite{Birk}), we define
\[
R_{t}=\left(  \widetilde{\Sigma}^{T}(t)\widetilde{\Sigma}(t)\right)
^{-1/2}\widetilde{\Sigma}(t)\in U(n)
\]
hence $R_{t}$ is of the type%
\[
R_{t}=%
\begin{pmatrix}
A_{t} & B_{t}\\
-B_{t} & A_{t}%
\end{pmatrix}
\text{ \ , }u_{t}=A_{t}+iB_{t}\in U(n,\mathbb{C}).
\]
To the path $\widetilde{\Sigma}$ we associate a path $\gamma$ in $\mathbb{C}$
by the formula $\gamma(t)=(\det u_{t})^{2}$. We have $\gamma(0)=\gamma(2)=1$
and $|\gamma(t)|=1$ hence $\gamma$ is in fact a loop in $S^{1}$. The
Conley--Zehnder of $\Sigma$ is, by definition, the winding number of that
loop:%
\begin{equation}
i_{\mathrm{CZ}}(\Sigma)=\frac{1}{2\pi i}%
{\displaystyle\oint\nolimits_{\gamma}}
\frac{dz}{z}. \label{deficz}%
\end{equation}
We note that since $\operatorname*{Sp}\nolimits_{+}(n)$ and
$\operatorname*{Sp}\nolimits_{-}(n)$ are contractible, the integer
$i_{\mathrm{CZ}}(\Sigma)$ does not depend on the choice of the extension
$\widetilde{\Sigma}$ of the symplectic path $\Sigma$.

The main properties of the Conley--Zehnder index are summarized below. We
denote by $\mathcal{C(}\operatorname*{Sp}_{\pm}(n))$ the set of all symplectic
isotopies having their endpoint in $\operatorname*{Sp}_{\pm}%
(n)=\operatorname*{Sp}_{+}(n)\cup\operatorname*{Sp}_{-}(n)$. We denote by
$\Sigma\ast\Sigma^{\prime}$ the concatenation of two paths $\Sigma$ and
$\Sigma^{\prime}$.

\begin{proposition}
The index $i_{\mathrm{CZ}}$ is the only mapping $\mathcal{C(}%
\operatorname*{Sp}_{\pm}(n))\longrightarrow\mathbb{Z}$ having the following properties:

(\textbf{CZ1}) The integer $i_{\mathrm{CZ}}(\Sigma)$ only depends on the
homotopy class (with fixed endpoints) of $\Sigma$;

(\textbf{CZ2}) For every $\Sigma\in\operatorname*{Sp}_{\pm}(n)$ we have
$i_{\mathrm{CZ}}(\Sigma^{-1})=-i_{\mathrm{CZ}}(\Sigma)$;

(\textbf{CZ3}) Let $\Sigma\in\operatorname*{Sp}_{\pm}(n)$ have endpoint $S$
and let $\Sigma^{\prime}$ be a continuous path joining $S$ to $S^{\prime}$ in
the same connected component $\operatorname*{Sp}_{+}(n)$ or
$\operatorname*{Sp}_{-}(n)$ as $S$. Then $i_{\mathrm{CZ}}(\Sigma\ast
\Sigma^{\prime})=i_{\mathrm{CZ}}(\Sigma)$;

(\textbf{CZ4}) For every $r\in\mathbb{Z}$ we have $i_{\mathrm{CZ}}(\Sigma
\ast\alpha^{r})=i_{\mathrm{CZ}}(\Sigma)+2r$ ($\alpha$ the generator of
$\pi_{1}[\operatorname*{Sp}(n)]\equiv\pi_{1}[U(n)]$ whose image in
$\mathbb{Z}$ is $+1$).
\end{proposition}

\begin{proof}
See \cite{Hofer} and \cite{Birk}, \S 4.3.1. Observe that (\textbf{CZ1}) is an
immediate consequence of the definition (\ref{deficz}) of the Conley--Zehnder
index since $\operatorname*{Sp}\nolimits_{+}(n)$ and $\operatorname*{Sp}%
\nolimits_{-}(n)$ are contractible.
\end{proof}

We moreover have the following important conjugation property:

\begin{proposition}
Let $\Sigma=(S_{t})_{t\in I}$ be a symplectic isotopy in $\operatorname*{Sp}%
(n)$ with endpoint $S\notin\operatorname*{Sp}_{0}(n)$.

(i) $i_{\mathrm{CZ}}(\Sigma)$ is locally constant on the set of all $\Sigma$
with fixed $\dim\operatorname*{Ker}(S-I)$;

(ii) For every $R\in\operatorname*{Sp}(n)$ we have%
\begin{equation}
i_{\mathrm{CZ}}(R\Sigma R^{-1})=i_{\mathrm{CZ}}(\Sigma) \label{fund}%
\end{equation}
where $R\Sigma R^{-1}=(RS_{t}R^{-1})_{t\in I}$.
\end{proposition}

\begin{proof}
See Meinrenken \cite{mein} (Proposition 6). Formula (\ref{fund}) follows from
(i) since $i_{\mathrm{CZ}}(R\Sigma R^{-1})$ is invariant if one connects $R$
to the identity in $\operatorname*{Sp}(n)$ (alternatively, it readily follows
from definition (\ref{deficz})).
\end{proof}

One of us has defined \cite{Birk,RMPCZ,useful} the Conley--Zehnder index of
symplectic path without any restriction on the endpoint using the properties
of the Leray index (see Appendix B), to which it is closely related.
Introducing the symplectic form $\sigma^{\ominus}=\sigma\oplus(-\sigma)$ on
$\mathbb{R}^{2n}\oplus\mathbb{R}^{2n}\equiv\mathbb{R}^{2n}\times
\mathbb{R}^{2n}$ and denoting by $\operatorname*{Sp}\nolimits^{\ominus}(2n)$
and $\operatorname*{Lag}^{\ominus}(2n)$ the corresponding symplectic group and
Grassmannian Lagrangian, the following results identifies the Conley--Zehnder
index as previously defined with the Leray index:

\begin{proposition}
\label{PropLCZ}Let $\Sigma=(S_{t})_{t\in I}$ be an arbitrary symplectic
isotopy in $\operatorname*{Sp}(n)$ with endpoint $S$. Let%
\begin{equation}
\nu(\Sigma)=\frac{1}{2}\mu_{\Delta}^{\ominus}((I\oplus S)_{\infty})
\label{defcz}%
\end{equation}
where $\mu_{\Lambda}^{\ominus}$ is the $\Delta$-Maslov index on the universal
covering group $\operatorname*{Sp}\nolimits_{\infty}^{\ominus}(2n)$ of
$\operatorname*{Sp}\nolimits^{\ominus}(2n)$ with $\Delta=\{(z,z):z\in
\mathbb{R}^{2n}\}$ and $(I\oplus S)_{\infty}\in\operatorname*{Sp}%
\nolimits_{\infty}^{\ominus}(2n)$ the homotopy class of the path
\[
I\ni t\longmapsto\{(z,S_{t}z):z\in\mathbb{R}^{2n}\}\in\operatorname*{Lag}%
\nolimits^{\ominus}(2n).
\]
We have $\nu(\Sigma)\in\mathbb{Z}$ and
\begin{equation}
\nu(\Sigma)=i_{\mathrm{CZ}}(\Sigma) \label{eqCZ}%
\end{equation}
when $S\notin\operatorname*{Sp}_{0}(n)$.
\end{proposition}

\begin{proof}
That $\frac{1}{2}\mu_{\Delta}^{\ominus}((I\oplus S)_{\infty})\in\mathbb{Z}$
can be seen as follows: in view of the congruence in (\ref{mll'}) we have%
\[
\mu_{\Delta}^{\ominus}((I\oplus S)_{\infty})\equiv2n+\dim((I\oplus
S)\Delta\cap\Delta)\text{ \ }\operatorname{mod}2.
\]
But
\[
(I\oplus S)\Delta\cap\Delta=\{z\in\mathbb{R}^{2n}:Sz=z\}=\operatorname*{Ker}%
(S-I)
\]
so that
\[
\mu_{\Delta}^{\ominus}((I\oplus S)_{\infty})=2n+\dim\operatorname*{Ker}%
(S-I)\text{ \ }\operatorname{mod}2.
\]
The eigenvalue $1$ of a symplectic mapping having even multiplicity the
integer $\dim\operatorname*{Ker}(S-I)$ is always even and so is thus
$\mu_{\Delta}^{\ominus}((I\oplus S)_{\infty})$. Using the characteristic
property (MA) of the relative Maslov index (Appendix B2) it is easy to show
that the restriction of $\nu$ to $\operatorname*{Sp}_{+}(n)\cup
\operatorname*{Sp}_{-}(n)$ satisfies the properties (CZ1)--(CZ4) of the
Conley--Zehnder index (for a detailed argument see \cite{Birk}, \S 4.3.3 or
\cite{useful}. We thus have $\nu(\Sigma)=i_{\mathrm{CZ}}(\Sigma)$ for such paths.
\end{proof}

To study the Conley--Zehnder index of products of symplectic isotopies we need
the notion of symplectic Cayley transform of $S\in\operatorname*{Sp}_{0}(n)$.
It is, by definition \cite{Birk,RMPCZ}, the symmetric $2n\times2n$ matrix%
\begin{equation}
M(S)=\frac{1}{2}J(S+I)(S-I)^{-1}=\frac{1}{2}J+J(S-I)^{-1}. \label{symcay1}%
\end{equation}
It has the following properties:
\begin{equation}
M(S^{-1})=-M(S)\text{ \ and \ }R^{T}M(S)R=M(R^{-1}SR) \label{conjcay}%
\end{equation}
for $S\notin\operatorname*{Sp}\nolimits_{0}(n)$ and $R\in\operatorname*{Sp}%
(n)$. We will use several times in this paper following addition result:

\begin{lemma}
\label{leminv}Let $S,S^{\prime}\in\operatorname*{Sp}_{0}(n)$. If $SS^{\prime
}\in\operatorname*{Sp}_{0}(n)$ then $M=M(S)+M(S^{\prime})$ is invertible and
we have%
\begin{equation}
M=J(S-I)^{-1}(SS^{\prime}-I)(S^{\prime}-I)^{-1}. \label{MJ}%
\end{equation}

\end{lemma}

\begin{proof}
See \cite{Birk}, \S 4.3.2, Lemma 4.1.4.
\end{proof}

In the following result we give an explicit expression for the Conley--Zehnder
index of the product of two symplectic isotopies:

\begin{proposition}
Let $\Sigma=(S_{t})_{t\in I}$ and $\Sigma^{\prime}=(S_{t}^{\prime})_{t\in I}$
be two symplectic isotopies and set $\Sigma\Sigma^{\prime}=(S_{t}S_{t}%
^{\prime})_{t\in I}$ with endpoints $S$ and $S^{\prime}$ not in
$\operatorname*{Sp}\nolimits_{0}(n)$. If $SS^{\prime}\notin\operatorname*{Sp}%
\nolimits_{0}(n)$ then
\begin{equation}
\nu(\Sigma\Sigma^{\prime})=\nu(\Sigma)+\nu(\Sigma^{\prime})+\tfrac{1}%
{2}\operatorname*{sign}(M(S)+M(S^{\prime})) \label{modprod}%
\end{equation}
where $M(S)$ and $M(S^{\prime})$ are the symplectic Cayley transforms of $S$
and $S^{\prime}$, and $\operatorname*{sign}M$ is the signature of the
invertible symmetric matrix $M$.
\end{proposition}

\begin{proof}
See \cite{Birk} (Proposition 4.20) or \cite{RMPCZ} pp. 1163--1164 for detailed
proofs. Notice that $\operatorname*{sign}M=\operatorname*{sign}%
(M(S)+M(S^{\prime}))$ is an even integer since $M$ is invertible.
\end{proof}

\subsection{Metaplectic isotopies\label{sec22}}

The metaplectic group $\operatorname*{Mp}(n)$ being a twofold covering of the
symplectic group (see Appendix A), it follows from the path lifting property
for covering groups that every symplectic isotopy $\Sigma=(S_{t})_{t\in I}$ in
$\operatorname*{Sp}(n)$ can be lifted in a unique way to a path
$\widehat{\Sigma}:t\longmapsto\widehat{S}_{t}$ ($t\in I$) in
$\operatorname*{Mp}(n)$ such that $\widehat{S}_{0}=I_{\mathrm{d}}$
(\textit{i.e.} $\pi^{\operatorname*{Mp}}(\widehat{S}_{t})=S_{t}$). Conversely,
every such path (which we call a metaplectic isotopy) covers a symplectic
isotopy. The following result is well-known \cite{Folland,Birk}:

\begin{proposition}
Let $\widehat{\Sigma}=(\widehat{S}_{t})_{t\in I}$ be a metaplectic isotopy and
$\Sigma=(S_{t})_{t\in I}$ the symplectic isotopy it covers. We have%
\[
i\hbar\frac{d}{dt}\widehat{S}_{t}=\widehat{H}\widehat{S}_{t}%
\]
where $H$ is the quadratic Hamiltonian function (\ref{hamzo}) determined by
$(S_{t})_{t}$ and $\widehat{H}$ is the Weyl quantization of $H${. }
\end{proposition}

This result thus \emph{justifies} Schr\"{o}dinger's equation in the case of
quadratic Hamiltonians; for a detailed discussion of the quantum-classical
correspondence from the symplectic point of view see de Gosson \cite{RMP}.

We will need the following elementary conjugation result in our calculation of
the relative phase shift:

\begin{lemma}
\label{Lemmacov}Let $\Sigma^{H}=(S_{t}^{H})_{t\in I}$ be the symplectic
isotopy determined by a Hamiltonian $H$. (i) For $R\in\operatorname*{Sp}(n)$
the symplectic isotopy $R\Sigma^{H}R^{-1}=(RS_{t}^{H}R^{-1})_{t\in I}$ is
determined by $H\circ R^{-1}$, that is
\[
R\Sigma^{H}R^{-1}=\Sigma^{H\circ R^{-1}}.
\]
(ii) Let $\widehat{\Sigma}^{H}=(\widehat{S}_{t}^{H})_{t\in I}$ be the
metaplectic isotopy induced by $\Sigma^{H}$. We have $\widehat{R}%
\widehat{\Sigma}^{H}\widehat{R}^{-1}=\widehat{\Sigma}^{H\circ R^{-1}}$ that
is
\[
(\widehat{R}\widehat{S}_{t}^{H}\widehat{R}^{-1})_{t\in I}=(\widehat{S}%
_{t}^{H\circ R^{-1}})_{t\in I}%
\]
where $\pi^{\operatorname*{Mp}}(\widehat{R})=R$.
\end{lemma}

\begin{proof}
Property (i) follows from formula (\ref{hamzo}) noting that $JS=(S^{T})^{-1}J$
since $S\in\operatorname*{Sp}(n)$. Property (ii) follows from (i) using the
symplectic covariance of Weyl operators which implies that
\[
i\hbar\frac{d}{dt}(\widehat{R}\widehat{S}_{t}^{H}\widehat{R}^{-1}%
)=\widehat{R}^{-1}(\widehat{H}\widehat{S}_{t}^{H})\widehat{R}=(\widehat{R}%
^{-1}\widehat{H}\widehat{R})(\widehat{R}^{-1}\widehat{S}_{t}^{H}\widehat{R});
\]
noting that $\widehat{R}\widehat{H}\widehat{R}^{-1}=H\circ R^{-1}$ we have
$\widehat{S}_{t}^{H\circ R^{-1}}=\widehat{R}\widehat{S}_{t}^{H}\widehat{R}%
^{-1}$ in view of the uniqueness of the solution to Schr\"{o}dinger's equation
with given initial datum in $\mathcal{S}(\mathbb{R}^{n})$.
\end{proof}

We define the (extended) Conley--Zehnder index $\nu(\widehat{\Sigma})$ of a
metaplectic isotopy $\widehat{\Sigma}=(\widehat{S}_{t})_{t\in I}$ on
$\operatorname*{Mp}(n)$ as being $[\nu(\Sigma)]_{\operatorname{mod}4}$, the
class modulo $4$ of the Conley--Zehnder index $\nu(\Sigma)$ of the projected
symplectic path $\Sigma=\pi^{\operatorname*{Mp}}(\widehat{\Sigma})$. When the
endpoint $\widehat{S}$ of $\widehat{\Sigma}$ is such that $S=\pi
^{\operatorname*{Mp}}(\widehat{S})\notin\operatorname*{Sp}_{0}(n)$ then, by
Proposition \ref{PropLCZ},%
\begin{equation}
\nu(\widehat{\Sigma})=[i_{\mathrm{CZ}}(\Sigma)]_{\operatorname{mod}4}.
\label{iczmod4}%
\end{equation}

Following result is important for practical calculations; it shows that the
value modulo four of the Conley--Zehnder index $i_{\mathrm{CZ}}(\Sigma)$ of a
symplectic isotopy with endpoint a free symplectic matrix ~$S_{W}$ is related
to the Maslov index of the corresponding metaplectic isotopy:

\begin{proposition}
Let $\Sigma=(S_{t})_{t\in I}$ be a symplectic isotopy in $\operatorname*{Sp}%
(n)$ with endpoint $S\notin\operatorname*{Sp}\nolimits_{0}(n)$ and $S\ell
_{P}\cap\ell_{P}=0$. Thus $S=S_{W}$ for some generating function $W$ (see
(\ref{PLQ})). We have%
\begin{equation}
\nu(\widehat{\Sigma})=i_{\mathrm{CZ}}(\widehat{\Sigma}%
)=m-\operatorname*{Inert}W_{xx}\text{ \ }\operatorname{mod}4 \label{mod4}%
\end{equation}
where $\widehat{S}=\widehat{S}_{W,m}$ is the endpoint of the metaplectic
isotopy $\widehat{\Sigma}=(\widehat{S}_{t})_{t\in I}$.
\end{proposition}

\begin{proof}
Formula (\ref{mod4}) follows from the definition (\ref{defcz}) of the
Conley--Zehnder index and the properties of the Leray index.
\end{proof}

\begin{remark}
The integer $\operatorname*{Inert}W_{xx}$ in appearing (\ref{mod4}) is called
\textquotedblleft Morse's index of concavity\textquotedblright\ \cite{Morse}
in the literature on periodic Hamiltonian orbits.
\end{remark}

In terms on the explicit expression (\ref{PLQ}) of the quadratic form $W$,
that is%
\[
W(x,x^{\prime})=\tfrac{1}{2}Px^{2}-Lx\cdot x^{\prime}+\tfrac{1}{2}Qx^{\prime2}%
\]
we have%
\[
W_{xx}=P-L-L^{T}+Q.
\]
Thus, if the metaplectic isotopy $\widehat{\Sigma}$ has endpoint
$\widehat{S}_{W,m}$ we have%
\[
\nu(\widehat{\Sigma})=m-\operatorname*{Inert}(P-L-L^{T}+Q).
\]

\section{Metaplectic Operators and their Weyl Symbols\label{secmeta}}

The Weyl symbol of a metaplectic was introduced implicitly and without
justification in the work \cite{MW} of Mehlig and Wilkinson, and was studied
rigorously in \cite{LMP} (also see \cite{RMPCZ,Birk}).

\subsection{Weyl operators}

We will make use of the two following elementary but yet very useful unitary
operators on $L^{2}(\mathbb{R}^{n})$. For $z_{0}=(x_{0},p_{0})\in
L^{2}(\mathbb{R}^{2n})$ the displacement (or Weyl--Heisenberg) operator
$\widehat{T}(z_{0})$ and the reflection (or Grossmann--Royer
\cite{Grossmann,Royer,Birk,Wigner}) operator $\widehat{\Pi}(z_{0})$ are
defined by%
\begin{align}
\widehat{T}(z_{0})\psi(x)  &  =e^{\frac{i}{\hbar}(p_{0}x-\frac{1}{2}p_{0}%
x_{0})}\psi(x-x_{0})\label{HW}\\
\widehat{\Pi}(z_{0})\psi(x)  &  =e^{\frac{2i}{\hbar}p_{0}(x-x_{0})}\psi
(2x_{0}-x). \label{GR}%
\end{align}
One verifies that $\widehat{\Pi}(z_{0})=\widehat{T}(z_{0})\widehat{\Pi
}(0)\widehat{T}(z_{0})^{\ast}$ and that
\begin{align}
\widehat{T}(z_{0}+z_{1})  &  =e^{-\tfrac{i}{2\hslash}\sigma(z_{0},z_{1}%
)}\widehat{T}(z_{0})\widehat{T}(z_{1})\label{tzotzo}\\
\widehat{T}(z_{0})\widehat{T}(z_{1})  &  =e^{\tfrac{i}{\hslash}\sigma
(z_{0},z_{1})}\widehat{T}(z_{1})\widehat{T}(z_{0}). \label{tzotzobis}%
\end{align}
For $a\in L^{1}(\mathbb{R}^{2n})$ the function $F_{\sigma}a$ defined by
\begin{equation}
F_{\sigma}a(z)=\left(  \tfrac{1}{2\pi\hbar}\right)  ^{n}\int e^{-\frac
{i}{\hbar}\sigma(z,z^{\prime})}a(z^{\prime})d^{2n}z^{\prime} \label{SFT}%
\end{equation}
is called the symplectic Fourier transform of $a$; we have $F_{\sigma
}a(z)=Fa(Jz)$ where $F$ is the usual Fourier transform on $\mathbb{R}^{2n}$.
The operators $\widehat{T}(z_{0})$ and $\widehat{\Pi}(z_{0})$ are related by
$F_{\sigma}$ in the following way: \ \ \ \
\begin{equation}
\widehat{\Pi}(z_{0})\psi=2^{-n}F_{\sigma}[\widehat{T}(\cdot)\psi](-z_{0})
\label{defgr}%
\end{equation}
(see \cite{Birkbis}, \S 8.3.3, Proposition 149). The displacement and
reflection operators $\widehat{\Pi}(z_{0})$ allow to give very simple
definitions of the cross-Wigner and cross-ambiguity functions
\cite{Grossmann,Royer,Birk}: for $\psi,\phi\in\mathcal{S}(\mathbb{R}^{n})$
they are defined \cite{Birk,Wigner} by, respectively,
\begin{align}
W(\psi,\phi)(z)  &  =\left(  \tfrac{1}{\pi\hbar}\right)  ^{n}(\widehat{\Pi
}(z_{0})\psi|\phi)\label{wpsifi}\\
A(\psi,\phi)(z)  &  =\left(  \tfrac{1}{2\pi\hbar}\right)  ^{n}(\psi
|\widehat{T}(z)\phi); \label{apsifi}%
\end{align}
a straightforward computation using the definitions (\ref{GR}) and (\ref{HW})
show that these equivalent to the usual expressions%
\begin{align}
W(\psi,\phi)(z)  &  =\left(  \tfrac{1}{2\pi\hbar}\right)  ^{n}\int
e^{-\frac{i}{\hbar}p\cdot y}\psi(x+\tfrac{1}{2}y)\overline{\phi(x-\tfrac{1}%
{2}y)}d^{n}y\label{wpsi}\\
A(\psi,\phi)(z)  &  =\left(  \tfrac{1}{2\pi\hbar}\right)  ^{n}\int
e^{-\tfrac{i}{\hbar}p\cdot y}\psi(y+\tfrac{1}{2}x)\overline{\phi(y-\tfrac
{1}{2}x)}d^{n}y. \label{apsi}%
\end{align}
The functions $W(\psi,\phi)$ and $A(\psi,\phi)$ are related by the symplectic
Fourier transform \cite{Birk}
\[
W(\psi,\phi)=F_{\sigma}A(\psi,\phi)\text{ \ , \ }A(\psi,\phi)=F_{\sigma}%
W(\psi,\phi)
\]
as immediately follows from formula (\ref{defgr}).

Let $a\in\mathcal{S}(\mathbb{R}^{2n})$. The Weyl operator with symbol $a$ is
defined by $\widehat{A}=\operatorname*{Op}\nolimits^{\mathrm{W}}(a)$ where%
\begin{equation}
\operatorname*{Op}\nolimits^{\mathrm{W}}(a)\psi(x)=\left(  \tfrac{1}{2\pi
\hbar}\right)  ^{n}\iint e^{\frac{i}{\hbar}p(x-y)}a(\tfrac{1}{2}%
(x+y),p)\psi(y)d^{n}pd^{n}y. \label{opw1}%
\end{equation}
Using the displacement operator this definition takes the more tractable form%
\begin{equation}
\operatorname*{Op}\nolimits^{\mathrm{W}}(a)=\left(  \tfrac{1}{\pi\hbar
}\right)  ^{n}\int a(z_{0})\widehat{\Pi}(z_{0})d^{2n}z_{0}. \label{opw3}%
\end{equation}
Setting $a_{\sigma}=F_{\sigma}a$ we also have, using the Plancherel theorem
for the symplectic Fourier transform \cite{Birk},%
\begin{equation}
\operatorname*{Op}\nolimits^{\mathrm{W}}(a)=\left(  \tfrac{1}{2\pi\hbar
}\right)  ^{n}\int a_{\sigma}(z_{0})\widehat{T}(z_{0})d^{2n}z_{0}.
\label{opw2}%
\end{equation}
The integrals above should be interpreted as vector-valued integrals (Bochner
integrals). Using the generators $\widehat{J}$, $\widehat{V}_{-P}$, and
$\widehat{M}_{L,m}$ defined by (\ref{mp1}), (\ref{mp2}), (\ref{mp3}) it is a
simple exercise to show that the displacement and reflection operators satisfy
the symplectic covariance relations
\begin{equation}
\widehat{T}(Sz_{0})=\widehat{S}\widehat{T}(z_{0})\widehat{S}^{-1}\text{
\ \textit{,} \ }\widehat{\Pi}(Sz_{0})=\widehat{S}\widehat{\Pi}(z_{0}%
)\widehat{S}^{-1} \label{cov}%
\end{equation}
for every $\widehat{S}\in\operatorname*{Mp}(n)$, $S=\pi^{\operatorname*{Mp}%
}(\widehat{S})$, and using (\ref{opw2}) it follows that Weyl operators satisfy
the conjugation formula%
\begin{equation}
\widehat{S}\operatorname*{Op}\nolimits^{\mathrm{W}}(a)\widehat{S}%
^{-1}=\operatorname*{Op}\nolimits^{\mathrm{W}}(a\circ S^{-1}). \label{sympcov}%
\end{equation}
We also have a conjugation formula for the operators $\widehat{T}(z_{0})$:
\begin{equation}
\widehat{T}(z_{0})\operatorname*{Op}\nolimits^{\mathrm{W}}(a)\widehat{T}%
(z_{0})^{-1}=\operatorname*{Op}\nolimits^{\mathrm{W}}(T(z_{0})a)
\label{tzocov}%
\end{equation}
where by definition $T(z_{0})a(z)=a(z-z_{0})$; this easily follows from
formula (\ref{tzotzobis}).

Formula (\ref{wpsi}) implies that we have the following relation:%
\begin{equation}
\langle\operatorname*{Op}\nolimits^{\mathrm{W}}(a)\psi,\overline{\phi}%
\rangle=\langle\langle a,W(\psi,\phi)\rangle\rangle\label{average}%
\end{equation}
for all $\psi,\phi\in\mathcal{S}(\mathbb{R}^{n})$; here $\langle\cdot
,\cdot\rangle$ and $\langle\langle\cdot,\cdot\rangle\rangle$ are the
distributional brackets on $\mathbb{R}^{n}$ and $\mathbb{R}^{2n}$,
respectively. This formula allows to define $\widehat{A}=\operatorname*{Op}%
\nolimits^{\mathrm{W}}(a)$ for arbitrary symbols $a\in\mathcal{S}^{\prime
}(\mathbb{R}^{2n})$ \cite{Birk,Wigner}.

\begin{proposition}
\label{propcomp}Assume that the product $\widehat{A}\widehat{B}%
=\operatorname*{Op}\nolimits^{\mathrm{W}}(a)\operatorname*{Op}%
\nolimits^{\mathrm{W}}(b)$ is well defined. Then the twisted Weyl symbol of
$\widehat{C}=\widehat{A}\widehat{B}$ is given by the \textquotedblleft twisted
convolution\textquotedblright\ formula%
\begin{equation}
c_{\sigma}(z)=\left(  \tfrac{1}{2\pi\hbar}\right)  ^{n}\int e^{\frac{i}%
{2\hbar}\sigma(z,z^{\prime})}a_{\sigma}(z-z^{\prime})b_{\sigma}(z^{\prime
})d^{2n}z^{\prime}. \label{cecomp}%
\end{equation}
This in particular applies when $a$ or $b$ is in $\mathcal{S}(\mathbb{R}%
^{2n})$.
\end{proposition}

\begin{proof}
See for instance \cite{Birk}, \S 6.3.2.
\end{proof}

\subsection{The Weyl symbol of a quadratic Fourier transform}

Let us denote by $\operatorname*{Sp}\nolimits_{0}(n)$ the subset of
$\operatorname*{Sp}(n)$ consisting of all symplectic matrices having no
eigenvalue equal to one:
\begin{equation}
\operatorname*{Sp}\nolimits_{0}(n)=\{S\in\operatorname*{Sp}(n):\det
(S-I)\neq0\}. \label{A9}%
\end{equation}

Let now $S\in\operatorname*{Sp}\nolimits_{0}(n)$ and consider the family of
operators $\widehat{R}_{\nu}(S)$ defined, for $\nu\in\mathbb{R}$, by
\begin{equation}
\widehat{R}_{\nu}(S)=\left(  \tfrac{1}{2\pi\hbar}\right)  ^{n}i^{\nu}%
\sqrt{|\det(S-I)|}\int\widehat{T}(Sz_{0})\widehat{T}(-z_{0})d^{2n}z_{0}.
\label{rus1}%
\end{equation}
One verifies that for all $S\in\operatorname*{Sp}\nolimits_{0}(n)$ and $\nu
\in\mathbb{R}$ the operators $\widehat{R}_{\nu}(S)$ satisfy the intertwining
formula
\[
\widehat{T}(Sz_{0})=\widehat{R}_{\nu}(S)\widehat{T}(z_{0})\widehat{R}_{\nu
}(S)^{-1}.
\]
It follows, using the irreducibility of the Schr\"{o}dinger representation of
the Heisenberg group \cite{Folland}, that there exists a constant $c(S,\nu
)\in\mathbb{C}$ such that $\widehat{R}_{\nu}(S) = c(S,\nu)\widehat{S}$ where
$\pi^{\operatorname*{Mp}}(\widehat{S})=S$. It is moreover easy to check that
the operators are $\widehat{R}_{\nu}(S)$ unitary, hence $|c(S,\nu)|=1$. The
following result connects the integer $\nu$ in (\ref{rus1}) to the
Conley--Zehnder index when $\widehat{R}_{\nu}(S)$ is a true metaplectic operator:

\begin{proposition}
\label{propcz}Let $\Sigma=(S_{t})_{t\in I}$ be symplectic isotopy in
$\operatorname*{Sp}(n)$ leading from the identity to $S\notin%
\operatorname*{Sp}\nolimits_{0}(n)$. Let $\widehat{\Sigma}=(\widehat{S}%
_{t})_{t\in I}$ be the metaplectic isotopy covering $\Sigma$ and
$\widehat{S}\in\operatorname*{Mp}(n)$ be its endpoint (thus $S=\pi
^{\operatorname*{Mp}}(\widehat{S})$). We have
\[
\widehat{S}=\widehat{R}_{\nu(\widehat{\Sigma})}(S)
\]
where $\nu(\widehat{\Sigma})=\nu(\Sigma)$ $\operatorname{mod}4$.
\end{proposition}

\begin{proof}
This results from the identity (\ref{iczmod4}) (see \cite{LMP} and
\cite{RMPCZ,Birk}).
\end{proof}

The statement above has the following consequences when the endpoint of the
symplectic isotopy $\Sigma$ is a free symplectic matrix $S_{W}$:

\begin{corollary}
Let $\widehat{S}_{W,m}\in\operatorname*{Mp}(n)$ be such that $S_{W}%
=\pi^{\operatorname*{Mp}}(\widehat{S}_{W,m})\notin\operatorname*{Sp}%
\nolimits_{0}(n)$. We then have
\begin{equation}
\widehat{S}_{W,m}=\widehat{R}_{m-\operatorname*{Inert}W_{xx}}(S) \label{swmnu}%
\end{equation}
where $\operatorname*{Inert}W_{xx}$ is the index of inertia of the matrix
$W_{xx}$ of second derivatives of the quadratic form $x\longmapsto W(x,x)$ on
$\mathbb{R}^{n}$.
\end{corollary}

This allows us to give a rigorous explicit formula for the twisted Weyl symbol
of $\widehat{S}_{W,m}$:

\begin{corollary}
The twisted Weyl symbol of $\widehat{S}_{W,m}$ with $S_{W}\notin%
\operatorname*{Sp}\nolimits_{0}(n)$ is given by%
\begin{equation}
(s_{W})_{\sigma}(z)=\frac{i^{m-\operatorname*{Inert}W_{xx}}}{\sqrt{|\det
(S_{W}-I)|}}\exp\left(  \frac{i}{2\hbar}M_{W}z\cdot z\right)  \label{weylmp1}%
\end{equation}
where $M_{W}$ is the symplectic Cayley transform of $S_{W}$.
\end{corollary}

\begin{proof}
In view of (\ref{rus1}) and (\ref{swmnu}) we have
\[
\widehat{S}_{W,m}=\left(  \tfrac{1}{2\pi\hbar}\right)  ^{n}%
i^{m-\operatorname*{Inert}W_{xx}}\sqrt{|\det(S_{W}-I)|}\int\widehat{T}%
(S_{W}z_{0})\widehat{T}(-z_{0})d^{2n}z;
\]
using formula (\ref{tzotzo}) this can be rewritten
\[
\widehat{S}_{W,m}=\left(  \tfrac{1}{2\pi\hbar}\right)  ^{n}%
i^{m-\operatorname*{Inert}W_{xx}}\sqrt{|\det(S_{W}-I)|}\int e^{-\frac
{i}{2\hbar}\sigma(S_{W}z_{0},z_{0})}\widehat{T}((S_{W}-I)z_{0})d^{2n}z_{0}.
\]
Making the change of variable $z_{0}\longmapsto(S_{W}-I)^{-1}z_{0}$ we get
\[
\widehat{S}_{W,m}=\left(  \frac{1}{2\pi\hbar}\right)  ^{n}\frac
{i^{m-\operatorname*{Inert}W_{xx}}}{\sqrt{|\det(S_{W}-I)|}}\int e^{-\frac
{i}{2\hbar}\sigma(S_{W}(S_{W}-I)^{-1}z_{0},(S_{W}-I)^{-1}z_{0})}%
\widehat{T}(z_{0})d^{2n}z_{0}%
\]
hence the twisted Weyl symbol of $\widehat{S}_{W,m}$ is
\[
(s_{W})_{\sigma}(z)=\frac{i^{m-\operatorname*{Inert}W_{xx}}}{\sqrt{|\det
(S_{W}-I)|}}e^{-\frac{i}{2\hbar}\sigma(S_{W}(S_{W}-I)^{-1}z_{0},(S_{W}%
-I)^{-1}z_{0})}.
\]
A simple algebraic calculation shows that
\[
\sigma(S_{W}(S_{W}-I)^{-1}z_{0},(S_{W}-I)^{-1}z_{0})=\frac{1}{2}%
J(S_{W}+I)(S_{W}-I)^{-1}=M(S_{W});
\]
formula (\ref{weylmp1}) follows.
\end{proof}

Proposition \ref{propcz} and formula (\ref{weylmp1}) suggest that the
Conley--Zehnder index is related to a choice of argument of the square root of
the determinant of $S-I$. This is indeed the case:

\begin{proposition}
\label{propsw}Let $\widehat{S}_{W,m}\in\operatorname*{Mp}(n)$ have projection
$S_{W}\notin\operatorname*{Sp}_{0}(n)$. We have%
\begin{equation}
\nu(\widehat{S}_{W,m})=n+\frac{1}{\pi}\arg\det(S_{W}-I)\text{ \ }%
\operatorname{mod}2. \label{argdet1}%
\end{equation}
that is%
\begin{equation}
\nu(\widehat{S}_{W,m})=\left\{
\begin{array}
[c]{c}%
n\text{ \ }\operatorname{mod}2\text{ \ if }S_{W}\in\operatorname*{Sp}%
\nolimits_{+}(n)\\
n+2\text{ \ }\operatorname{mod}2\text{ \ if }S_{W}\in\operatorname*{Sp}%
\nolimits_{-}(n)
\end{array}
\right.  . \label{argdet2}%
\end{equation}

\end{proposition}

\begin{proof}
The projection $S_{W}=\pi^{\operatorname*{Mp}}(\widehat{S}_{W,m})$ is a free
symplectic matrix, in block-matrix form%
\[
S_{W}=%
\begin{pmatrix}
A & B\\
C & D
\end{pmatrix}
\text{ \ , \ }\det B\neq0.
\]
A straightforward calculation yields the factorization
\[
S_{W}-I=%
\begin{pmatrix}
0 & B\\
I & D-I
\end{pmatrix}%
\begin{pmatrix}
C-(D-I)B^{-1}(A-I) & 0\\
B^{-1}(A-I) & I
\end{pmatrix}
.
\]
Since $S_{W}\in\operatorname*{Sp}(n)$ we have $C-DB^{-1}A=-(B^{T})^{-1}$ and
hence
\[
C-(D-I)B^{-1}(A-I)=B^{-1}A+DB^{-1}-(B^{T})^{-1}=W_{xx}%
\]
so that
\[
S_{W}-I=%
\begin{pmatrix}
0 & B\\
I & D-I
\end{pmatrix}%
\begin{pmatrix}
W_{xx} & 0\\
B^{-1}(A-I) & I
\end{pmatrix}
.
\]
It follows that
\[
\det(S_{W}-I)=(-1)^{n}\det B\det W_{xx}%
\]
and hence%
\[
\arg\det(S_{W}-I)=n\pi+\arg\det B+\arg\det W_{xx}\text{ \ }\operatorname{mod}%
2\pi.
\]
Noticing that $\arg\det W_{xx}=\pi\operatorname{Inert}W_{xx}$ and that this is%
\[
\arg\det(S_{W}-I)=n\pi+\arg\det B+\pi\operatorname{Inert}W_{xx}\text{
\ }\operatorname{mod}2\pi.
\]
In view of formula (\ref{mod4}) and (\ref{weylmp1}) we have $\arg\det(B) =
m\pi$ (see Appendix A) and hence%
\[
\arg\det(S_{W}-I)=(n+m-\operatorname{Inert}W_{xx})\pi\text{ \ }%
\operatorname{mod}2\pi
\]
that is
\[
\arg\det(S_{W}-I)=(n+\nu(\widehat{S}_{W,m}))\pi\text{ \ }\operatorname{mod}%
2\pi
\]
which yields (\ref{argdet1}).
\end{proof}

\subsection{Products of metaplectic operators}

Each $\widehat{S}\in\operatorname*{Mp}(n)$ can be written as a product
$\widehat{S}_{W,m}\widehat{S}_{W^{\prime},m^{\prime}}$ (Appendix A,
Proposition \ref{propA1}). It turns out that $\widehat{S}_{W,m}$ and
$\widehat{S}_{W^{\prime},m^{\prime}}$ can be chosen so that their projections
$S_{W}$ and $S_{W^{\prime}}$ have no eigenvalue equal to one. This fact,
together with the composition formula\ (\ref{cecomp}) leads to a complete
characterization of the symbol of a metaplectic operator. When $\widehat{S}$
has projection $S\notin\operatorname*{Sp}\nolimits_{0}(n)$ we have the
following explicit result:

\begin{proposition}
Let $\widehat{S}\in\operatorname*{Mp}(n)$ be such that $\pi
^{\operatorname*{Mp}}(\widehat{S})\notin\operatorname*{Sp}\nolimits_{0}(n)$.

(i) There exist $\widehat{S}_{W,m}$ and $\widehat{S}_{W^{\prime},m^{\prime}}$
such that $\widehat{S}=\widehat{S}_{W,m}\widehat{S}_{W^{\prime},m^{\prime}}$;
moreover these operators can be chosen so that $S_{W}=\pi^{\operatorname*{Mp}%
}(\widehat{S}_{W,m})\notin\operatorname*{Sp}\nolimits_{0}(n)$ and
$S_{W^{\prime}}=\pi^{\operatorname*{Mp}}(\widehat{S}_{W^{\prime},m^{\prime}%
})\notin\operatorname*{Sp}\nolimits_{0}(n)$.

(ii) We have%
\begin{equation}
\widehat{S}=\widehat{R}_{\nu+\nu^{\prime}+\frac{1}{2}\operatorname*{sign}%
(M)}(S)=\widehat{R}_{\nu(\widehat{S})}(S) \label{sr}%
\end{equation}
where $M=M_{W}+M_{W^{\prime}}$ ($M_{W}$ and $M_{W^{\prime}}$ the symplectic
Cayley transforms of $S_{W}$ and $S_{W^{\prime}}$), and%
\begin{equation}
\nu=m-\operatorname*{Inert}W_{xx}\text{ \ , \ }\nu^{\prime}=m^{\prime
}-\operatorname*{Inert}W_{xx}^{\prime} \label{smm}%
\end{equation}
are the Conley--Zehnder indices of $\widehat{S}_{W,m}$ and $\widehat{S}%
_{W^{\prime},m^{\prime}}$;

(iii) The twisted Weyl symbol of $\widehat{S}$ is given by
\begin{equation}
s_{\sigma}(z)=\frac{i^{\nu(\widehat{S})}}{\sqrt{|\det(S-I)|}}\exp\left(
\frac{i}{2\hbar}Mz\cdot z\right)  \label{productmw}%
\end{equation}
with
\begin{equation}
\nu(\widehat{S})=\nu+\nu^{\prime}+\tfrac{1}{2}\operatorname*{sign}M.
\label{nus}%
\end{equation}

\end{proposition}

\begin{proof}
See Proposition 10 in \cite{LMP} or \cite{Birk}, \S 7.4 for detailed proofs.
That $\widehat{S}$ can always be factored as $\widehat{S}_{W,m}\widehat{S}%
_{W^{\prime},m^{\prime}}$ where $\widehat{S}_{W,m}$ and $\widehat{S}%
_{W^{\prime},m^{\prime}}$ have projections $S_{W}$ and $S_{W^{\prime}}$ not in
$\operatorname*{Sp}\nolimits_{0}(n)$ was proven in \cite{LMP}. For formula
(\ref{productmw}) the idea is to apply formula (\ref{cecomp}) to
(\ref{weylmp1}) and to use the Fresnel formula (\ref{Fresnel}), which yields,
after some calculations%
\[
c_{\sigma}(z)=\frac{i^{\nu+\nu^{\prime}+\frac{1}{2}\operatorname*{sign}(M)}%
}{\sqrt{|\det[(S_{W}-I)(S_{W^{\prime}}-I)M]|}}e^{\frac{i}{2\hbar}Mz\cdot z}.
\]
A simple calculation taking into account the definition of the symplectic
Cayley transforms shows that%
\begin{equation}
(S_{W}-I)(S_{W^{\prime}}-I)M=S-I \label{identity}%
\end{equation}
($M$ is invertible in view of Lemma \ref{leminv}).
\end{proof}

We have seen in Proposition \ref{propsw} that the Conley--Zehnder index of a
quadratic Fourier transform $\widehat{S}_{W,m}$ is simply related to a choice
of argument for $\det(S_{W}-I)$. Using the result above, this observation can
be generalized to the case of an arbitrary $\widehat{S}\in\operatorname*{Mp}%
(n)$ with projection $S\notin\operatorname*{Sp}\nolimits_{0}(n)$:

\begin{corollary}
Let $\widehat{S}\in\operatorname*{Mp}(n)$ with $S=\pi^{\operatorname*{Mp}%
}(\widehat{S})\notin\operatorname*{Sp}\nolimits_{0}(n)$. We have%
\begin{equation}
\nu(\widehat{S})=n+\frac{1}{\pi}\operatorname*{Arg}\det(S-I)\text{
\ }\operatorname{mod}2. \label{argdet3}%
\end{equation}

\end{corollary}

\begin{proof}
Writing $\widehat{S}=\widehat{S}_{W,m}\widehat{S}_{W^{\prime},m^{\prime}}$
with $S_{W}$ and $S_{W^{\prime}}$ not in $\operatorname*{Sp}\nolimits_{0}(n)$
it follows from the identity (\ref{identity}) that
\[
\det\left[  (S_{W}-I)(S_{W^{\prime}}-I)M\right]  =\det(S-I)
\]
with $M=M_{W}+M_{W^{\prime}}$ and hence
\[
\arg\det(S-I)=\arg\det(S_{W}-I)+\arg\det(S_{W^{\prime}}-I)+\arg\det M .
\]
Since $M=M_{W}+M_{W^{\prime}}$ is invertible (Lemma \ref{leminv}) we have
\[
\arg\det M=\pi\operatorname{Inert}M=-\pi\operatorname{Inert}M\text{
\ }\operatorname{mod}2\pi
\]
and hence, using formulas (\ref{argdet1}) and (\ref{modprod}) together with
the relation $\operatorname*{sign}M=2(n-\operatorname{Inert}M)$,
\begin{align*}
\arg\det(S-I)  &  =\nu(\widehat{S}_{W,m})\pi+\nu(\widehat{S}_{W,m}%
)-\pi(n-\tfrac{1}{2}\operatorname*{sign}M)\text{\ \ }\operatorname{mod}2\pi\\
&  =\nu(\widehat{S}_{W,m})\pi+\nu(\widehat{S}_{W,m})-n\pi+\tfrac{1}{2}%
\pi\operatorname*{sign}M\text{\ \ }\operatorname{mod}2\pi\\
&  =\nu(\widehat{S})\pi-n\pi\text{ \ }\operatorname{mod}2\pi
\end{align*}
proving formula (\ref{argdet3}).
\end{proof}

\section{Gaussian Density Operators\label{secdensity}}

\subsection{Generalities}

A density operator on a complex Hilbert space $\mathcal{H}$ is a positive
semidefinite (and hence selfadjoint) trace class operator $\widehat{\rho}$ on
$\mathcal{H}$ with unit trace. Every trace class operator is the product of
two Hilbert--Schmidt operators, and is hence compact. The spectral theorem for
compact operators implies that there exists a family $(\psi_{j})$ of
orthonormal vectors in $\mathcal{H}$ such that
\begin{equation}
\widehat{\rho}=\sum_{j}\lambda_{j}\widehat{\Pi}_{j} \label{spec1}%
\end{equation}
where $\widehat{\Pi}_{j}$ is the orthogonal projection on the vector $\psi
_{j}$; the $\lambda_{j}$ are the eigenvalues corresponding to the eigenvectors
$\psi_{j}$ and we have $\operatorname*{Tr}(\widehat{\rho})=\sum_{j}\lambda
_{j}=1$. In what follows we will assume that $\mathcal{H}=L^{2}(\mathbb{R}%
^{n})$. Let $a$ be the Weyl symbol of $\widehat{\rho}$; by definition
$\rho=(2\pi\hbar)^{-n}a$ is the Wigner distribution of $\widehat{\rho}$.
Taking (\ref{spec1}) into account we have $\rho=\sum_{j}\lambda_{j}W\psi_{j}$
where $W\psi_{j}$ is the Wigner transform of $\psi_{j}$.

Let $\widehat{A}=\operatorname*{Op}\nolimits^{\mathrm{W}}(a)$ and
$\widehat{B}=\operatorname*{Op}\nolimits^{\mathrm{W}}(b)$ be two trace class
operators (or, more generally, Hilbert--Schmidt operators). Since
Hilbert-Schmidt operators are precisely those with kernels in $L^{2}%
(\mathbb{R}^{n}\times\mathbb{R}^{n})$ it follows that $a\in L^{2}%
(\mathbb{R}^{2n})$ and $b\in L^{2}(\mathbb{R}^{2n})$. The product
$\widehat{A}\widehat{B}$ is of trace class and we have
\begin{equation}
\operatorname*{Tr}(\widehat{A}\widehat{B})=\left(  \tfrac{1}{2\pi\hbar
}\right)  ^{n}\int a(z)b(z)d^{2n}z. \label{trab1}%
\end{equation}
Let $a_{\sigma}=F_{\sigma}a$ be the symplectic Fourier transform of $a$ (see
(\ref{SFT})). The formula%
\begin{equation}
\operatorname*{Tr}(\widehat{A})=\left(  \tfrac{1}{2\pi\hbar}\right)  ^{n}\int
a(z)d^{2n}z=a_{\sigma}(0) \label{tra}%
\end{equation}
is often used in the literature; one should however be aware that it is only
true if one assumes that in addition $a\in L^{1}(\mathbb{R}^{2n})$ (see
\cite{Brislawn,duwong}). (It is instructive to read B. Simon's analysis in
\cite{BSimon} of trace formulas of this type; also see Shubin \cite{sh87}, \S 27).

\subsection{Gaussian states}

Let $\widehat{\rho}$ be a density operators whose Wigner distribution is a
Gaussian:%
\begin{equation}
\rho(z)=(2\pi)^{-n}\sqrt{\det V^{-1}}e^{-\frac{1}{2}V^{-1}z\cdot z}.
\label{rhosig2}%
\end{equation}
The covariance matrix $V$ is a positive definite symmetric (real) $2n\times2n$
matrix, and $z=(x,p)$ is in the phase space $\mathbb{R}^{2n}\equiv
\mathbb{R}^{n}\times\mathbb{R}^{n}$. A necessary and sufficient condition for
a function (\ref{rhosig2}) to represent a quantum state is that the Hermitian
matrix $V+(i\hbar/2)J$ ($J$ the standard symplectic matrix) has no negative
eigenvalues; for short%
\begin{equation}
V+\frac{i\hbar}{2}J\geq0\text{.} \label{quantum1}%
\end{equation}
This condition ensures that the density operator $\widehat{\rho}$ is indeed
positive semidefinite and is equivalent in the Gaussian case to the
Robertson--Schr\"{o}dinger uncertainty principle: see de Gosson and Luef
\cite{goluPR}). It will be convenient to set $V=\frac{1}{2}\hbar F^{-1}$; with
that notation the Wigner distribution of the Gaussian state (\ref{rhosig2})
can be written
\begin{equation}
\rho(z)=(\pi\hbar)^{-n}\sqrt{\det F}e^{-\frac{1}{\hbar}Fz\cdot z}
\label{rhof1}%
\end{equation}
and the quantum condition (\ref{quantum1}) becomes $F^{-1}+iJ\geq0$. The
symplectic Fourier transform (\ref{SFT}) of $\rho$ is given by{
\begin{equation}
\rho_{\sigma}(z)=(2\pi\hbar)^{-n}e^{\frac{1}{4\hbar}JF^{-1}Jz\cdot z}%
=(2\pi\hbar)^{-n}e^{-\frac{1}{4\hbar}F^{-1}Jz\cdot Jz}. \label{rhof1sig}%
\end{equation}
}Notice that if $F=F^{T}\in\operatorname*{Sp}(n)$ then $JFJ=-F^{-1}$ hence, in
this case{
\begin{equation}
\rho_{\sigma}(z)=(2\pi\hbar)^{-n}e^{-\frac{1}{4\hbar}Fz\cdot z} \label{little}%
\end{equation}
which is the symplectic Fourier transform of the Wigner transform of a
generalized coherent state} \cite{Littlejohn,Birk} (see below).

The purity of $\widehat{\rho}$ is by definition $\mu=\operatorname*{Tr}%
(\widehat{\rho}^{2})$, and we have in the Gaussian case
\begin{equation}
\mu=\left(  \frac{\hbar}{2}\right)  ^{n}\det(V)^{-1/2}=\sqrt{\det F}.
\label{puresigma}%
\end{equation}
Thus $\widehat{\rho}$ is a pure state ($\mu=1$) if and only if $\det
(V)=(\hbar/2)^{2n}$, that is $\det F=1$. One shows \cite{Littlejohn,Birk} that
this equivalent to the existence of $R\in\operatorname*{Sp}(n)$ such that
$V=R^{T}R$. It follows that the only Gaussian pure states are those with
Wigner distribution%
\begin{equation}
\rho(z)=(\pi\hbar)^{-n}e^{-\tfrac{1}{\hbar}S^{T}Sz\cdot z}.
\end{equation}

Let $\phi_{0}$ be the standard coherent state: $\phi_{0}(x)=(\pi\hbar
)^{-n/4}e^{-|x|^{2}/2\hbar}$. The action of the local metaplectic operators
$\widehat{M}_{L,m}$ and $\widehat{V}_{-P}$ on $\phi_{0}$ is the $L^{2}%
$-normalized Gaussian $\psi_{X,Y}=\widehat{V}_{-P}\widehat{M}_{L,m}\phi_{0}$
given by%
\[
\psi_{X,Y}(x)=i^{m}(\pi\hbar)^{-n/4}(\det X)^{1/4}e^{-\tfrac{1}{2\hbar
}(X+iY)x\cdot x}%
\]
where $X=L^{T}L$ and $Y=P$. The Wigner transform of $\psi_{X,Y}$ is
(\cite{Littlejohn}, \cite{Birk} \S 8.5):%
\begin{equation}
W\psi_{X,Y}(z)=(\pi\hbar)^{-n}e^{-\tfrac{1}{\hbar}Gz\cdot z} \label{wigauss1}%
\end{equation}
where $G$ is the positive definite symplectic and symmetric matrix%
\begin{equation}
G=%
\begin{pmatrix}
X+YX^{-1}Y & YX^{-1}\\
X^{-1}Y & X^{-1}%
\end{pmatrix}
. \label{gsym}%
\end{equation}
The ambiguity function $F_{\sigma}\rho_{X,Y}$ is easily calculated and one
finds that%
\begin{equation}
F_{\sigma}\rho_{X,Y}(z)=\left(  \tfrac{1}{2\pi\hbar}\right)  ^{n}e^{-\tfrac
{1}{4\hbar}Gz^{2}}. \label{ambgauss1}%
\end{equation}

\section{Relative Phase Shifts\label{secpancha}}

\subsection{A general result}

We will need the following generalization of formula (\ref{trab1}):

\begin{lemma}
\label{Lemmatrab}Let $\widehat{S}\in\operatorname*{Mp}(n)$ be such that
$\pi^{\operatorname*{Mp}}(\widehat{S})\notin\operatorname*{Sp}\nolimits_{0}%
(n)$ and $\rho$ the Gaussian distribution (\ref{rhof1}). The product
$\widehat{S}\widehat{\rho}$ is of trace class and we have
\begin{equation}
\operatorname*{Tr}(\widehat{S}\widehat{\rho})=\left(  \tfrac{1}{2\pi\hbar
}\right)  ^{n}\int s_{\sigma}(z)\rho_{\sigma}(z)d^{2n}z \label{trab3}%
\end{equation}
where $s_{\sigma}$ is the twisted Weyl symbol of $\widehat{S}$.
\end{lemma}

\begin{proof}
The product $\widehat{C}=\widehat{S}\widehat{\rho}$ is of trace class because
trace class operators form a two-sided ideal in the algebra of bounded
operators on $L^{2}(\mathbb{R}^{n})$. We can however not apply directly
formula (\ref{trab1}) since $\widehat{S}$ is not a Hilbert--Schmidt operator.
Let us proceed as follows: in view of formula (\ref{cecomp}) in Proposition
\ref{propcomp}. The twisted Weyl symbol $c_{\sigma}$ of $\widehat{C}$ is given
by the absolutely convergent integral
\[
c_{\sigma}(z)=\left(  \tfrac{1}{2\pi\hbar}\right)  ^{n}\int e^{-\frac
{i}{2\hbar}\sigma(z,z^{\prime})}s_{\sigma}(z^{\prime})\rho_{\sigma
}(z-z^{\prime})d^{2n}z^{\prime}%
\]
where the twisted symbol $s_{\sigma}$ of $\widehat{S}$ is given by formula
(\ref{productmw}) and $\rho_{\sigma}$ is the Gaussian (\ref{rhof1sig}). Since
$\rho_{\sigma}\in\mathcal{S}(\mathbb{R}^{2n})$ we have $c_{\sigma}%
\in\mathcal{S}(\mathbb{R}^{2n})$ and hence also $c\in\mathcal{S}%
(\mathbb{R}^{2n})$, and we may therefore apply the trace formula (\ref{tra})
which yields%
\[
\operatorname*{Tr}(\widehat{S}\widehat{\rho})=\operatorname*{Tr}%
(\widehat{C})=c_{\sigma}(0)
\]
which is precisely formula (\ref{trab3}).
\end{proof}

\begin{notation}
\label{remnotation}To simplify the statements below it will be convenient to
write $\nu(S_{t})$ for the Conley--Zehnder index of the symplectic isotopy
$t^{\prime}\longmapsto S_{t^{\prime}}$ for $0\leq t^{\prime}\leq t$ and
$\nu(\widehat{S}_{t})$ for the Conley--Zehnder index of the corresponding
metaplectic isotopy.
\end{notation}

Using this notation we have:

\begin{theorem}
\label{Thm1} Let $\Sigma=(S_{t})_{t\in I}$ be a symplectic isotopy with
endpoint $S\in$. Let $\widehat{\Sigma}=(\widehat{S}_{t})_{t\in I}$ be the
associated metaplectic isotopy, and assume that $S_{t}\notin\operatorname*{Sp}%
\nolimits_{0}(n)$. We have%
\begin{equation}
\operatorname*{Tr}(\widehat{S}_{t}\widehat{\rho})=\frac{i^{\nu(\widehat{S}%
_{t})}}{\sqrt{|\det(S_{t}-I)|}}\det\nolimits^{-1/2}(\tfrac{1}{2}%
F^{-1}+iM(S_{t}^{T})). \label{formess1}%
\end{equation}
The relative phase shift is thus given by the formula%
\begin{equation}
\phi(t)=\frac{\pi}{2}\nu(\widehat{S}_{t})+\operatorname*{Arg}\det
\nolimits^{-1/2}(\tfrac{1}{2}F^{-1}+iM(S_{t}^{T})). \label{phi1}%
\end{equation}

\end{theorem}

\begin{proof}
We have $\rho\in L^{1}(\mathbb{R}^{2n})\cap L^{2}(\mathbb{R}^{2n})$ and
$\widehat{S}_{t}$ is bounded on $L^{2}(\mathbb{R}^{n})$. Let $s_{t}$ be the
Weyl symbol of $\widehat{S}_{t}$; applying formula (\ref{trab3}) in Lemma
\ref{Lemmatrab} we get
\[
\operatorname*{Tr}(\widehat{S}_{t}\widehat{\rho})=\int F_{\sigma}%
s_{t}(z)F_{\sigma}\rho(-z)d^{2n}z.
\]
Since $F_{\sigma}s_{t}=(s_{t})_{\sigma}$ is the twisted symbol of the
metaplectic operator $\widehat{S}_{t}$ (formula (\ref{weylmp1})) we get,
taking the expression (\ref{rhof1sig}) into account
\[
\operatorname*{Tr}(\widehat{S}_{t}\widehat{\rho})=\frac{(2\pi\hbar)^{-n}%
i^{\nu(\widehat{S}_{t})}}{\sqrt{|\det(S_{t}-I)|}}\int e^{-\frac{1}{2\hbar
}(-\frac{1}{2}JF^{-1}J-iM_{t})z^{2}}d^{2n}z.
\]
Taking $A=-\frac{1}{2}JFJ-iM_{t}$ in the Fresnel formula (\ref{Fresnel}), we
get%
\[
\operatorname*{Tr}(\widehat{S}_{t}\widehat{\rho})=\frac{i^{\nu(\widehat{S}%
_{t})}}{\sqrt{|\det(S_{t}-I)|}}\det\nolimits^{-1/2}\left(  \tfrac{1}{2}%
F^{-1}+iJM(S_{t})J\right)
\]
hence the result once $\nu(\widehat{S}_{t})=\nu_{\mathrm{CZ}}(\Sigma)$ modulo
four, since $JM(S_{t})J=M(S_{t}^{T})$ in view of the second formula
(\ref{conjcay}).
\end{proof}

\subsection{Application: harmonic oscillator and standard coherent state}

Assume that the symplectic path $\Sigma$ consists of the rotations%
\[
S_{t}=%
\begin{pmatrix}
\cos\omega t & \sin\omega t\\
-\sin\omega t & \cos\omega t
\end{pmatrix}
.
\]
Then $H=\frac{\omega}{2}(x^{2}+p^{2})$ and we have for $\omega t\notin%
\pi\mathbb{Z}$ and $\psi_{0}\in\mathcal{S}(\mathbb{R})$
\begin{equation}
\widehat{S}_{t}\psi_{0}(x)=i^{-[\omega t/\pi]}\sqrt{\frac{1}{2\pi i\hbar
|\sin\omega t|}}\int_{-\infty}^{\infty}e^{\frac{i}{\hbar}W(x,x^{\prime}%
,t)}\psi_{0}(x^{\prime})dx^{\prime} \label{st1}%
\end{equation}
where $m(\widehat{S}_{t})=-[\omega t/\pi]$ is the usual Maslov index
($[\alpha]$ denotes the integer part of $\alpha\in\mathbb{R}$); the generating
function is here%
\begin{equation}
W(x,x^{\prime},t)=\frac{1}{2\sin\omega t}\left[  (x^{2}+x^{\prime2})\cos\omega
t-2xx^{\prime}\right]  \label{stw}%
\end{equation}
(see \textit{e.g.} \cite{dire}, \textit{pp}.196--198). One verifies by direct
calculation that the function $\psi(\cdot,t)=\widehat{S}_{t}\psi_{0}$
satisfies%
\[
i\hbar\frac{\partial\psi}{\partial t}=\frac{\omega}{2}\left(  -\hbar^{2}%
\frac{\partial^{2}}{\partial x^{2}}+x^{2}\right)  \psi\text{ \ , \ }\psi
(\cdot,0)=\psi_{0}.
\]

Choose $\rho(z)=(\pi\hbar)^{-1}e^{-|z|^{2}/\hbar}$ (it is the Wigner transform
of the standard coherent state $\phi_{0}(x)=(\pi\hbar)^{-4}e^{-x^{2}/2\hbar}%
$). We have here $W_{xx}=-2\tan(\omega t/2)$. Also,
\begin{equation}
M_{t}=\frac{1}{2}%
\begin{pmatrix}
\cot(\omega t/2) & 0\\
0 & \cot(\omega t/2)
\end{pmatrix}
\label{mt}%
\end{equation}
hence, since $F=I$ in this case,
\begin{equation}
\det\left(  \tfrac{1}{2}I+iJM_{t}J\right)  =\frac{-e^{i\omega t}}{4\sin
^{2}(\omega t/2)}. \label{detmt}%
\end{equation}
Using the prescriptions following Fresnel's formula (\ref{Fresnel}) we get,
setting $A_{t}=\tfrac{1}{2}I+iJM_{t}J$,%
\[
\det\nolimits^{-1/2}A_{t}=\sqrt{-4\sin^{2}(\omega t/2)e^{i\omega t}}%
\]
and hence, writing $\operatorname*{Arg}\det\nolimits^{-1/2}A_{t}%
=\operatorname*{Arg}(t)$,
\[
\operatorname*{Arg}(t)=\left\{
\begin{array}
[c]{c}%
-\dfrac{\omega t+\pi}{2}\text{ \ \textit{for} \ }2k\pi<\omega t<(2k+1)\pi\\
\dfrac{\omega t-\pi}{2}\text{ \ \textit{for} \ }(2k+1)\pi<\omega t<2(k+1)\pi.
\end{array}
\right.
\]
On the other hand, using formula (\ref{mod4}) we have
\begin{equation}
\nu(\widehat{S}_{t})=-[\frac{\omega t}{\pi}]-\operatorname*{Inert}\left(
-\tan(\frac{\omega t}{2})\right)  \text{ \ }\operatorname{mod}4 \label{nust}%
\end{equation}
where $\operatorname*{Inert}\alpha=0$ if $\alpha>0$ and $\operatorname*{Inert}%
\alpha=1$ if $\alpha<0$; explicitly:%
\[
\nu(\widehat{S}_{t})=-2(k+1)\text{ for }2k\pi<\omega t<2(k+1)\pi.
\]

Using formula (\ref{formess1}) in Theorem \ref{Thm1}, the phase $\varphi
(t)=\operatorname*{Arg}\operatorname*{Tr}(\widehat{S}_{t}\widehat{\rho})$ is
given by%
\[
\varphi(t)=-\frac{\pi}{2}\left(  [\frac{\omega t}{\pi}]+\operatorname*{Inert}%
\left(  -\tan(\frac{\omega t}{2})\right)  \right)  +\operatorname*{Arg}(t)
\]
and hence, summarizing:

\begin{center}%
\begin{tabular}
[c]{||l||l||l||l||}\hline\hline
$\ \ \ \ \ \ \ \ \ \ \ \ \ \omega t$ & $\ \ \nu(\widehat{S}_{t})$ &
$\ \operatorname*{Arg}(t)$ & $\ \ \ \ \ \ \ \ \ \ \ \varphi(t)$\\\hline\hline
$(2k\pi,(2k+1)\pi)$ & $-(2k+1)$ & $-\dfrac{\omega t+\pi}{2}$ & $2k\pi
-\dfrac{\omega t}{2}$\\\hline\hline
$((2k+1)\pi,2(k+1)\pi)$ & $-(2k+1)$ & $\dfrac{\omega t-\pi}{2}$ &
$2k\pi+\dfrac{\omega t}{2}$\\\hline\hline
\end{tabular}

\end{center}

which coincides with the results obtained by one of us in \cite{nicacio1}.

\subsection{The generalized oscillator}

We now consider symplectic isotopies associated with Hamiltonian functions of
the type%
\begin{equation}
H(z)=\frac{1}{2}Kz\cdot z \label{H}%
\end{equation}
where $K=K(t)$ is a positive definite symmetric real matrix. We recall
Williamson's symplectic diagonalization theorem (Folland \cite{Folland}, Ch.4
and \cite{Birk}, \S 8.3.1): there exists $R\in\operatorname*{Sp}(n)$ such that%
\begin{equation}
K=R^{T}DR\text{ \ , \ }D=%
\begin{pmatrix}
\Omega & 0\\
0 & \Omega
\end{pmatrix}
\label{lemfol}%
\end{equation}
where $\Omega$ is a diagonal matrix whose diagonal entries $\omega_{j}>0$ are
such that the $\pm i\omega_{j}$ are the eigenvalues of $JK$. The numbers
$\omega_{j}$ are the symplectic eigenvalues of $F$. We have%
\begin{equation}
H(R^{-1}z)=\frac{1}{2}Dz\cdot z=\sum_{j=1}^{n}\frac{\omega_{j}}{2}(x_{j}%
^{2}+p_{j}^{2}). \label{HR}%
\end{equation}
Rearranging the phase space coordinates by replacing $z=(x,p)$ with
$u=(x_{1},p_{1},...,x_{n},p_{n})$ the symplectic flow $S_{t}^{H\circ R^{-1}}$
is thus given by $u(t)=S_{t}^{H\circ R^{-1}}u$ where $(x_{j}(t),p_{j}%
(t))=S_{t}^{(j)}(x_{j},p_{j})$ with
\[
S_{t}^{(j)}=%
\begin{pmatrix}
\cos\omega_{j}t & \sin\omega_{j}t\\
-\sin\omega_{j}t & \cos\omega_{j}t
\end{pmatrix}
.
\]
The corresponding generating function will thus be $W=\sum_{j=1}^{n}W_{j}$
where the $W_{j}$ are given by formula (\ref{stw}):
\begin{equation}
W_{j}(x_{j},x_{j}^{\prime},t)=\frac{1}{2\sin\omega_{j}t}\left[  (x_{j}%
^{2}+x_{j}^{\prime2})\cos\omega_{j}t-2x_{j}x_{j}^{\prime}\right]  . \label{wj}%
\end{equation}

\begin{theorem}
\label{Thm2}Let $\Sigma=(S_{t}^{H})$ be the symplectic isotopy determined by
(\ref{H}). Let $\widehat{\Sigma}=(\widehat{S}_{t}^{H})$ be the corresponding
metaplectic isotopy. Let $\widehat{\rho}$ a density matrix with Gaussian
Wigner distribution (\ref{rhof1}). We have, for $\omega_{j}t\notin%
\pi\mathbb{Z}$,%
\begin{equation}
\operatorname*{Tr}(\widehat{S}_{t}^{H}\widehat{\rho})=\operatorname*{Tr}%
(\widehat{S}_{t}^{H\circ R^{-1}}\widehat{\rho\circ R^{-1}}) \label{tratra}%
\end{equation}
and hence%
\begin{multline*}
\varphi(t)=-\frac{\pi}{2}\sum_{j=1}^{n}\left(  [\frac{\omega_{j}t}{\pi
}]+\operatorname*{Inert}\left(  -\tan(\frac{\omega_{j}t}{2})\right)  \right)
\\
+\operatorname*{Arg}\det\nolimits^{-1/2}(-\tfrac{1}{2}(R^{-1})^{T}%
FR^{T}+iM(RS_{t}R^{-1})).
\end{multline*}

\end{theorem}

\begin{proof}
Let $R$ be as in the Williamson diagonalization (\ref{lemfol}); we have
$S_{t}^{H}=R^{-1}S_{t}^{H\circ R^{-1}}R$ and $\widehat{S}_{t}^{H}%
=\widehat{R}^{-1}\widehat{S}_{t}^{H\circ R^{-1}}\widehat{R}$ (Lemma
\ref{Lemmacov}) and thus%
\[
\operatorname*{Tr}(\widehat{S}_{t}^{H}\widehat{\rho})=\operatorname*{Tr}%
(\widehat{R}^{-1}\widehat{S}_{t}^{H\circ R^{-1}}\widehat{R}\widehat{\rho
})=\operatorname*{Tr}(\widehat{S}_{t}^{H\circ R^{-1}}\widehat{R}\widehat{\rho
}\widehat{R}^{-1})
\]
hence (\ref{tratra}) follows in view of the symplectic covariance result
(\ref{sympcov}). Since $W=\sum_{j=1}^{n}W_{j}$ and $W_{j}$ being given by
(\ref{wj}), it follows from formulas (\ref{st1}) and (\ref{wj}) that we have,
for $\omega_{j}t\notin\pi\mathbb{Z}$,
\[
\widehat{S}_{t}^{H\circ R^{-1}}\psi(x)=\left(  \tfrac{1}{2\pi i\hbar}\right)
^{n/2}\Delta(W)\int%
{\textstyle\prod\nolimits_{j=1}^{n}}
(e^{\frac{i}{\hbar}W_{j}(x_{j},x_{j}^{\prime},t)})\psi(x^{\prime}%
)d^{n}x^{\prime}%
\]
where
\[
\Delta(W)=i^{m(t)}|\sin\omega_{1}t\cdot\cdot\cdot\sin\omega_{n}t|^{-1/2}%
\]
and $m(t)=-\sum_{j=1}^{n}[\omega_{j}t/\pi]$ is the Maslov index. We have, by
definition (\ref{rhof1}) of $\rho$ and recalling that $\det R=1$,
\[
\rho\circ R^{-1}(z)=(\pi\hbar)^{-n}\sqrt{\det(R^{-1})^{T}FR^{-1}}e^{-\frac
{1}{\hbar}(R^{-1})^{T}FR^{-1}z\cdot z}.
\]
On the other hand
\[
\nu(\widehat{S}_{t}^{H})=\nu(\widehat{R}^{-1}\widehat{S}_{t}^{H\circ R^{-1}%
}\widehat{R})=\nu(\widehat{S}_{t}^{H\circ R^{-1-}})
\]
(property (\ref{fund}) of the Conley--Zehnder index). We thus have (formula
(\ref{nust}))
\[
\nu(\widehat{S}_{t}^{H\circ R^{-1}})=\nu(\widehat{S}_{t}^{(1)})+\nu
(\widehat{S}_{t}^{(1)})+\cdot\cdot\cdot+\nu(\widehat{S}_{t}^{(n)})
\]
where
\begin{equation}
\nu(\widehat{S}_{t}^{H\circ R^{-1}})=-\left(  \sum_{j=1}^{n}[\frac{\omega
_{j}t}{\pi}]+\operatorname*{Inert}\left(  -\tan(\frac{\omega_{j}t}{2})\right)
\right)  \text{ \ \ }\operatorname{mod}4.
\end{equation}
On the other hand we have $M(R^{-1}S_{t}^{H\circ R^{-1}}R)=R^{T}%
M(S_{t}^{H\circ R^{-1}})R$ in view of the second formula (\ref{conjcay})\ and
hence%
\begin{align*}
\tfrac{1}{2}F^{-1}+iJM_{S_{t}^{H}}J  &  =\tfrac{1}{2}F^{-1}+iJM(R^{-1}%
S_{t}^{H\circ R^{-1}}R)J\\
&  =\tfrac{1}{2}F^{-1}+iJR^{T}M(S_{t}^{H\circ R^{-1}})RJ\\
&  =\tfrac{1}{2}F^{-1}+iR^{-1}JM(S_{t}^{H\circ R^{-1}})J(R^{T})^{-1}.
\end{align*}
We thus have%
\[
\det(\tfrac{1}{2}F^{-1}+iJM(S_{t}^{H})J)=\det(\tfrac{1}{2}RF^{-1}%
R^{T}+iJM(S_{t}^{H\circ R^{-1}})J)
\]
where
\begin{equation}
M_{t}^{j}=\frac{1}{2}%
\begin{pmatrix}
\cot(\omega_{j}t/2) & 0\\
0 & \cot(\omega_{j}t/2)
\end{pmatrix}
\end{equation}
and thus%
\[
\phi(t)=\frac{\pi}{2}\nu(\Sigma)+\operatorname*{Arg}\det\nolimits^{-1/2}%
(-\tfrac{1}{2}F+iM_{t}).
\]

\end{proof}

\section{The Inhomogeneous Case}

Our central result is placed in theorem \ref{Thm1} which now will be slightly
generalized to a larger class of Hamiltonian dynamics, besides the quadratic
term we will include affine transformations related to phase space
displacements. Despite its simplicity, the new form of the Hamiltonian is
widely used in the literature as an approximation for any Hamiltonian
dynamical system, see for instance \cite{corobert,Littlejohn,nicacio2}.

\subsection{The groups $\operatorname*{HSp}(n)$ and $\operatorname*{IMp}(n)$}

The inhomogeneous symplectic group $\operatorname*{ISp}(n)$ is the semi-direct
product \cite{Burdet,comanota}
\[
\operatorname*{ISp}(n)=\operatorname*{Sp}(n)\ltimes\mathrm{T}(2n)
\]
where $\mathrm{T}(2n)$ is the group of phase space translations $T(z_{0}%
):z\longmapsto z+z_{0}$. Its elements are the affine symplectomorphisms
$ST(z)$ (or $T(z)S$) where $S\in\operatorname*{Sp}(n)$ and $T(z)\in
\mathrm{T}(2n)$; note that
\[
ST(z)=T(Sz)S\text{ \ , \ }T(z)S=ST(S^{-1}z).
\]
The group law of $\operatorname*{ISp}(n)$ is given by%
\[
(S,z)(S^{\prime},z^{\prime})=(SS^{\prime},S^{\prime- 1}z+z^{\prime})
\]

More interesting is, in a sense, the group $\operatorname*{HSp}(n)$
(\cite{Folland,Burdet}; it is denoted by $\operatorname*{WSp}(n)$ in
\cite{Burdet}). It is defined as follows: Let $\mathrm{H}(2n)$ be the
Heisenberg group, that is $\mathbb{R}^{2n}\times S^{1}$ equipped with the
group law%
\[
(z,t)(z^{\prime},t^{\prime})=(z+z^{\prime},t+t^{\prime}+\tfrac{1}{2}%
\sigma(z,z^{\prime})).
\]
The symplectic group acts on $\mathrm{H}(2n)$ by $S(z,t)=(Sz,t)$ hence we can
form the semidirect product $\operatorname*{Sp}(n)\ltimes\mathrm{H}(2n)$. By
definition this group is $\operatorname*{HSp}(n)$ the group law being given by%
\[
(S,(z,t))(S^{\prime},(z^{\prime},t^{\prime}))= (SS^{\prime},(S^{\prime-
1}z,t)(z^{\prime},t)),t^{\prime}))).
\]

Let now $\widehat{T}:\mathrm{H}(2n)\longrightarrow\mathcal{U}(L^{2}%
(\mathbb{R}^{n}))$ be the Schr\"{o}dinger representation of $\mathrm{H}(2n)$
defined by $\widehat{T}(z,t)=e^{\frac{i}{\hbar}\gamma_{t}}\widehat{T}(z)$.

We similarly denote by $\operatorname*{IMp}(n)$ the group of unitary operators
on $L^{2}(\mathbb{R}^{n})$ generated by the the operators $\widehat{S}%
\in\operatorname*{Mp}(n)$ and $\widehat{T}(z,t)$, $z\in\mathbb{R}^{2n}$. It
follows from the symplectic covariance relations%
\[
\widehat{S}\widehat{T}(z,t)=\widehat{T}(Sz,t)\widehat{S}\text{ \ ,
\ }\widehat{T}(z,t)\widehat{S}=\widehat{S}\widehat{T}(S^{-1}z,t)
\]
that every element of $\operatorname*{IMp}(n)$ can be written in the form
$\widehat{S}\widehat{T}(z)$ or $\widehat{T}(z)\widehat{S}$. This is often
referred to as the extended metaplectic representation of $\operatorname*{HSp}%
(n)$; the projection $\pi^{\operatorname*{IMp}}:\operatorname*{IMp}%
(n)\longrightarrow\operatorname*{HSp}(n)$ is given by
\[
\pi^{\operatorname*{IMp}}(\widehat{T}(z,t)\widehat{S})=(S,z,t).
\]
Notice that if we restrict ourselves to the case $t=0$ this reduces to
\[
\pi^{\operatorname*{IMp}}(\widehat{T}(z)\widehat{S})=(S,z)\in
\operatorname*{ISp}(n).
\]

\subsection{Symplectic paths in $\operatorname*{ISp}(n)$}

Consider an affine metaplectic isotopy $(\widehat{U}_{t})_{t\in\mathbb{R}}$
where $\widehat{U}_{t}\in\operatorname*{IMp}(n)$ is of the type%
\begin{equation}
\widehat{U}_{t}=\widehat{T}(z_{t},\gamma_{t})\widehat{S}_{t}=e^{\frac{i}%
{\hbar}\gamma_{t}}\widehat{T}(z_{t})\widehat{S}_{t}; \label{ut1}%
\end{equation}
here $t\longmapsto z_{t}=(x_{t},p_{t})$ is a $C^{1}$ path in $\mathbb{R}^{2n}$
and $(\widehat{S}_{t})_{t\in\mathbb{R}}$ a metaplectic isotopy. The phase
$\gamma_{t}\in\mathbb{R}$ depends in a $C^{1}$ fashion on $t$. A
straightforward calculation taking into account the identity $\widehat{p}%
\widehat{T}(z_{t})-\widehat{T}(z_{t})\widehat{p}=p_{t}\widehat{T}(z_{t})$
(\textit{cf}. (\ref{tzocov})) yields%
\begin{equation}
i\hbar\frac{d}{dt}\widehat{T}(z_{t})=-(\tfrac{1}{2}\sigma(z_{t},\dot{z}%
_{t})+\sigma(\dot{z}_{t},\widehat{z}))\widehat{T}(z_{t}) \label{ni1}%
\end{equation}
where $\sigma(\dot{z}_{t},\widehat{z})$ is the operator $J\dot{z}_{t}%
\cdot\widehat{z}$ with $\widehat{z}\psi=(x\psi,-i\hbar\partial_{x}\psi)$;
equivalently%
\begin{equation}
i\hbar\frac{d}{dt}\widehat{T}(z_{t})=\left(  \sigma(\widehat{z}-z_{t},\dot
{z}_{t})+\tfrac{1}{2}\sigma(z_{t},\dot{z}_{t})\right)  \widehat{T}(z_{t}).
\label{ni2}%
\end{equation}
On the other hand one easily verifies that%
\[
i\hbar\frac{d}{dt}\widehat{S}_{t}=H(\widehat{z},t)\widehat{S}_{t}%
\]
where $H(\widehat{z},t)$ is the Weyl quantization of the Hamiltonian function
$H(z,t)$ defined by (\ref{hamzo}), that is
\[
H(\widehat{z},t)=-\frac{1}{2}J\dot{S}_{t}S_{t}^{-1}\widehat{z}\cdot
\widehat{z}.
\]
Collecting these results we see that $\widehat{U}_{t}$ satisfies the
Schr\"{o}dinger equation
\[
i\hbar\frac{d}{dt}\widehat{U}_{t}=\left(  -\dot{\gamma}_{t}-\tfrac{1}{2}%
\sigma(z_{t},\dot{z}_{t})-\sigma(\dot{z}_{t},\widehat{z})+H(\widehat{z}%
-z_{t},t)\right)  \widehat{U}_{t}.
\]
We next observe that the operator
\begin{equation}
\label{quadham}\widehat{H}_{z_{t}}=-\dot{\gamma}_{t}-\tfrac{1}{2}\sigma
(z_{t},\dot{z}_{t})-\sigma(\dot{z}_{t},\widehat{z})+H(\widehat{z}-z_{t},t)
\end{equation}
occurring in this equation is the Weyl quantization of the inhomogeneous
quadratic Hamilton function%
\[
H_{z_{t}}(z,t)=-\dot{\gamma}_{t}-\tfrac{1}{2}\sigma(z_{t},\dot{z}_{t}%
)-\sigma(\dot{z}_{t},z)+H(z-z_{t},t).
\]
The solutions of the associated Hamilton equations%
\[
\dot{z}=J\partial_{z}H_{z_{t}}(z,t)=\dot{z}_{t}+J\partial_{z}H(z-z_{t},t)
\]
are given by $z=u+z_{t}$ where $u$ is the solution of the Hamilton equations
for $H(z,t)$. Recalling that the flow determined by $H(z,t)$ is the symplectic
isotopy $(S_{t})_{t\in\mathbb{R}}$ we thus have%
\[
z(t)=S_{t}(z(0)-z_{0})+z_{t}.
\]

\subsection{Application to relative phase shifts}

We assume now that $\rho$ is a Gaussian centered at a point $\overline{z}%
\in\mathbb{R}^{2n}$:%
\begin{equation}
\rho(z)=(2\pi)^{-n}\sqrt{\det V^{-1}}e^{-\frac{1}{2}V^{-1}(z-\overline
{z})\cdot(z-\overline{z})} \label{rhoshift}%
\end{equation}
and that $(\widehat{U}_{t})_{t\in\mathbb{R}}$ is an affine metaplectic isotopy
given by (\ref{ut1}). We will use the following elementary result:

\begin{lemma}
\label{Lemmacomp}Let $\widehat{A}=\operatorname*{Op}\nolimits^{\mathrm{W}}(a)$
and $z_{0}\in\mathbb{R}^{2n}$. We have $\widehat{T}(z_{0})\widehat{A}%
=\operatorname*{Op}\nolimits^{\mathrm{W}}(c)$ where%
\begin{equation}
c_{\sigma}(z)=a_{\sigma}(z-z_{0})e^{-\tfrac{i}{2\hslash}\sigma(z,z_{0})}.
\label{casigma}%
\end{equation}

\end{lemma}

\begin{proof}
The twisted Weyl symbol of $\widehat{T}(z_{0})$ is given by $t_{\sigma
}(z)=(2\pi\hbar)^{n}\delta(z-z_{0})$; formula (\ref{casigma}) follows using
(\ref{cecomp}).
\end{proof}

\begin{theorem}
\label{Thm3}Let $\widehat{\rho}$ be the density operator with Wigner
distribution (\ref{rhoshift}) and $(\widehat{U}_{t})_{t\in\mathbb{R}}$ the
metaplectic isotopy defined by (\ref{ut1}). We have
\begin{align}
\label{formess2}\operatorname*{Tr}(\widehat{U}_{t}\widehat{\rho})  & = \frac{
i^{i_{\mathrm{CZ}}(\Sigma)} \, \mathrm{e}^{\frac{i}{\hbar} \gamma_{t} +
\frac{i}{\hbar}{ J z_{t}\cdot\bar z} -\frac{1}{2\hbar}{J z_{t}\cdot F^{-1} J
z_{t}} + \Phi(z_{t},\bar z)}} { \sqrt{|\det( S_{t} - I)|} \sqrt{\det(\tfrac
{1}{2}F^{-1} +i M(S_{t}^{T}) )} },\\
\Phi(z_{t},\bar z)  & = \frac{1}{8\hbar} \left[ (F^{-1} + iJ) J z_{t} - 2 i
\bar z \right]  \cdot\left[  \tfrac{1}{2}F^{-1} + i M(S_{t}^{T}) \right] ^{-1}
\left[ (F^{-1} + iJ) J z_{t} - 2 i \bar z \right]  .\nonumber
\end{align}

\end{theorem}

\begin{proof}
In view of formula (\ref{casigma}) in Lemma \ref{Lemmacomp} the twisted Weyl
symbol of $\widehat{T}(z_{t})\widehat{S}_{t}$ is the function $z\longmapsto
(s_{t})_{\sigma}(z-z_{t})e^{-\frac{i}{2\hbar}\sigma(z,z_{t})}$. Proceeding as
in the proof of Theorem \ref{Thm1} we have%
\begin{align*}
\operatorname*{Tr}(\widehat{T}(z_{t})\widehat{S}_{t}\widehat{\rho})  &
=\int(s_{t})_{\sigma}(z-z_{t})e^{-\frac{i}{2\hbar}\sigma(z,z_{t})}\rho
_{\sigma}(-z)d^{2n}z\\
&  =\int(s_{t})_{\sigma}(z)e^{-\frac{i}{2\hbar}\sigma(z,z_{t})}\rho_{\sigma
}(z_{t}-z)d^{2n}z.
\end{align*}
Using the Fresnel formula (\ref{Fresnel}) with $(s_{t})_{\sigma}(z)$ in
(\ref{weylmp1}) and $\rho_{\sigma}$, the symplectic Fourier transform
(\ref{SFT}) of $\rho$ in (\ref{rhoshift}), eq.(\ref{formess2}) follows.
\end{proof}

The relative phase shift (\ref{shift2}) for the Gaussian state in
(\ref{rhoshift}) subjected to the inhomogeneous dynamics in (\ref{ut1}) will
thus be (see Notation \ref{remnotation}):%

\begin{equation}
\phi(t)=\frac{\pi}{2}\nu(\widehat{S}_{t}) + \frac{1}{\hbar} \gamma_{t} +
\operatorname*{Arg} \Phi(z_{t},\bar z) + \operatorname*{Arg}\det
\nolimits^{-1/2}(\tfrac{1}{2}F^{-1}+iM(S_{t}^{T})). \label{phi2}%
\end{equation}
This formula reduces to the one in (\ref{phi1}) when $\overline{z} = z_{t} =
0$ and $\gamma_{t} = 0, \, \forall t \in I $.

\section{Discussion and Perspectives}

The proof of the general result in Theorem \ref{Thm1} heavily relies on the
fact that the integral giving the trace is easily calculable because the
integrand is a Gaussian and can, as such, be explicitly determined by a
Fresnel-type formula. This relative straightforwardness is due to the fact
that the twisted Weyl symbol of the unitary evolution operator $(\widehat{U}%
_{t})$ is here a family $(\widehat{S}_{t})$ of metaplectic operators and is
hence itself a (complex) Gaussian function, namely
\[
s_{\sigma}(z)=\frac{i^{\nu+\nu^{\prime}+\frac{1}{2}\operatorname*{sign}(M)}%
}{\sqrt{|\det(S-I)|}}\exp\left(  \frac{i}{2\hbar}M(S)z\cdot z\right)
\]
when $\det(S-I)\neq0$. It would of course be interesting (and even essential)
to extend Theorem \ref{Thm1} to more general situations. But even when
$\widehat{\rho}$ is still a Gaussian state we run into a major difficulty,
which is the determination of the Weyl symbol (twisted, or not) of a general
evolution operator $(\widehat{U}_{t})$. Very little is actually known; to the
best of our knowledge only sporadic attempts exist in the literature, and they
usually consist in using non-rigorous Feynman-type path integral methods.

Fortunately, semiclassical propagation methods are very well developed in the
context of the inhomogeneous dynamics presented in section \ref{tic}. In this
scenario, a generic analytic Hamiltonian function of $\hat{z}$ can be expanded
up to second order, always giving rise to a quadratic structure (\ref{quadham}%
), as proposed in \cite{Littlejohn}. The validity of this approximation is
guaranteed for a time interval limited by the very known Ehrenfest time
$\tau_{\mathrm{E}}\sim\mathrm{log}(\hbar^{-1})$ \cite{corobert}. This scheme
is very well suited for propagation of Gaussian states, since the operators
(\ref{ut1}) keep this set of states invariant. However, any quantum state can
be expanded as a superposition of (Gaussian) coherent states, thus this method
can be applied to the propagation of any initial state
\cite{corobert,Littlejohn}. An example for the propagation of states under a
non-linear (and classically chaotic) Hamiltonian dynamics is given in
\cite{nicacio2}.

{On the other side, the study of the Pancharatnam phase for states outside of
the Gaussian set is possible and can reveal new and interesting scenarios for
the phase behavior. For instance, coherent and incoherent superpositions of
Gaussian states \cite{nicacio2} can be studied quite directly using the tools
presented in this paper; the relation of the total phase for the interference
fringes has not yet been explored in the literature. Even for more general
states, the developed tools can also be applied using the Glauber--Sudarshan
representation \cite{g-srep}, which constitutes an expansion, in principle
written for any quantum state, in terms of the standard coherent state basis.
More generally, the notion of Gabor (or Weyl--Heisenberg) \ frame could be
useful in this context \cite{fagoro17}.}

In the same spirit of the generalization of the total phase (\ref{shift1})
defined in \cite{mukunda} to any density state (\ref{shift2}), a possible
generalization of the dynamical phase \cite{mukunda}
\[
\varphi_{\mathrm{d}}=\mathrm{Im}\int_{0}^{t}(\psi_{t^{\prime}},\dot{\psi
}_{t^{\prime}})dt^{\prime},
\]
can be given, introducing the more general quantity
\[
\varphi_{\mathrm{d}}^{\prime}=\int_{0}^{t}\mathrm{Tr}(\hat{\rho}\hat
{H}(t^{\prime}))dt^{\prime}.
\]
However, in \cite{mukunda} the geometric phase is defined to be $\varphi
-\varphi_{\mathrm{d}}$ for $\varphi$ in (\ref{shift1}). The association of
generalized geometric phase to $\varphi-\varphi_{\mathrm{d}}^{\prime}$ should
be carefully investigated, specially in what concerns its relation with the
Conley--Zehnder index. This index can be viewed as a geometric quantity
associated to paths on the space $L^{2}(\mathbb{R}^{n})$ connecting Gaussian
states. This idea will be developed in forthcoming work.

\section*{APPENDIX\ A: The Metaplectic Group}

For detailed studies of the symplectic group $\operatorname*{Sp}(n)$ see
\cite{Dutta,Birk}; the properties of the metaplectic are studied in
\cite{Folland,Birk}.

\subsection*{A.1 Definition}

The metaplectic group $\operatorname*{Mp}(n)$ is a unitary representation on
$L^{2}(\mathbb{R}^{n})$ of the double cover $\operatorname*{Sp}_{2}(n)$ of the
symplectic group $\operatorname*{Sp}(n)$. The simplest (but not necessarily
the most useful) way to describe $\operatorname*{Mp}(n)$ is to use its
elementary generators $\widehat{J}$, $\widehat{V}_{-P}$, and $\widehat{M}%
_{L,m}$ \cite{Folland,Birk}. Denoting by $\pi^{\operatorname*{Mp}}$ the
covering projection $\operatorname*{Mp}(n)\longrightarrow\operatorname*{Sp}%
(n)$ these operators and their projections are given by%
\begin{align}
\widehat{J}\psi(x)  &  =e^{-in\pi/4}F\psi(x)\text{ \ , \ }\pi
^{\operatorname*{Mp}}(\widehat{J})=J\tag{A1}\label{mp1}\\
\widehat{V}_{-P}\psi(x)  &  =e^{\frac{i}{2\hbar}Px^{2}}\psi(x)\text{ \ ,
\ }\pi^{\operatorname*{Mp}}(\widehat{V}_{-P})=V_{-P}\tag{A2}\label{mp2}\\
\widehat{M}_{L,m}\psi(x)  &  =i^{m}\sqrt{|\det L|}\psi(Lx)\text{ \ , \ }%
\pi^{\operatorname*{Mp}}(\widehat{M}_{L,m})=M_{L,m}. \tag{A3}\label{mp3}%
\end{align}
Here $F$ is the unitary $\hbar$-Fourier transform%
\[
F\psi(x)=\left(  \tfrac{1}{2\pi\hbar}\right)  ^{n}\int e^{-\frac{i}{\hbar
}xx^{\prime}}\psi(x^{\prime})d^{n}x^{\prime}%
\]
and $V_{-P}$ ($P=P^{T}$), $M_{L,m}$ ($\det L\neq0$) are the symplectic
matrices
\[
V_{-P}=%
\begin{pmatrix}
I & 0\\
P & I
\end{pmatrix}
\text{ \ , \ }M_{L,m}=%
\begin{pmatrix}
L^{-1} & 0\\
0 & L^{T}%
\end{pmatrix}
.
\]
The index $m$ in $\widehat{M}_{L,m}$ is an integer corresponding to a choice
of $\arg\det L$: $m$ is even if $\det L>0$ and odd if $\det L<0$. It is called
the \textit{Maslov index} of $\widehat{M}_{L,m}$.

\subsection*{A.2 Definition using quadratic Fourier transforms}

Let $P,Q\in\operatorname*{Sym}(n,\mathbb{R})$ and $L\in GL(n,\mathbb{R})$, and
let $W$ be the real quadratic form on $\mathbb{R}^{n}\times\mathbb{R}^{n}$
defined by%
\begin{equation}
W(x,x^{\prime})=\tfrac{1}{2}Px^{2}-Lx\cdot x^{\prime}+\tfrac{1}{2}Qx^{\prime
2}. \tag{A4}\label{PLQ}%
\end{equation}
To $W$ we associate \cite{Leray,Birk} an operator $\widehat{S}_{W,m}%
:\mathcal{S}(\mathbb{R}^{n})\longrightarrow\mathcal{S}(\mathbb{R}^{n})$ by the
formula%
\[
\widehat{S}_{W,m}\psi(x)=e^{-ni\pi/4}\left(  \tfrac{1}{2\pi\hbar}\right)
^{n/2}i^{m}\sqrt{|\det L|}\int e^{\frac{i}{\hbar}W(x,x^{\prime})}%
\psi(x^{\prime})d^{n}x^{\prime}%
\]
where the integer $m$ (which is only defined modulo 4) corresponds to a choice
or $\arg\det L$ as above. By definition, that integer is the Maslov index of
\textit{Maslov index} of $\widehat{S}_{W,m}$. One verifies by a simple
calculation that we have%

\begin{equation}
\widehat{S}_{W,m}=\widehat{V}_{-P}\widehat{M}_{L,m}\widehat{J}\widehat{V}_{-Q}
\tag{A5}\label{swm}%
\end{equation}
hence $\widehat{S}_{W,m}\in\operatorname*{Mp}(n)$ is a unitary operator on
$L^{2}(\mathbb{R}^{n})$. Using the formulas (\ref{mp1})--(\ref{mp3}) a simple
calculation shows that $S_{W}=\pi^{\operatorname*{Mp}}(\widehat{S}_{W,m})$ is
given by%
\begin{equation}
S_{W}=%
\begin{pmatrix}
L^{-1}Q & L^{-1}\\
PL^{-1}Q-L^{T} & L^{-1}P
\end{pmatrix}
. \tag{A6}\label{A6}%
\end{equation}

The operators $\widehat{S}_{W,m}$ are called quadratic Fourier transforms; one
easily \cite{Leray,Birk} verifies that%
\begin{equation}
(\widehat{S}_{W,m})^{-1}=\widehat{S}_{W^{\ast},m^{\ast}}\text{ \textit{with}
}W^{\ast}(x,x^{\prime})=-W(x^{\prime},x),\ m^{\ast}=n-m. \tag{A7}\label{A7}%
\end{equation}
The quadratic Fourier transforms form a dense subset of $\operatorname*{Mp}%
(n)$. In fact they generate this group:

\begin{proposition}
\label{propA1}Every $\widehat{S}\in\operatorname*{Mp}(n)$ can be written as a
the product of two quadratic Fourier transforms: $\widehat{S}=\widehat{S}%
_{W,m}\widehat{S}_{W^{\prime},m^{\prime}}$ and $\pi^{\operatorname*{Mp}%
}(\widehat{S}_{W,m})=S_{W}$ where $S_{W}\in\operatorname*{Sp}(n)$ is generated
by the quadratic form $W$, that is%
\[
(x,p)=S_{W}(x^{\prime},p^{\prime})\Longleftrightarrow\left\{
\begin{array}
[c]{c}%
p=\partial_{x}W(x,x^{\prime})\\
p^{\prime}=-\partial_{x^{\prime}}W(x,x^{\prime}).
\end{array}
\right.
\]

\end{proposition}

\begin{proof}
See Leray \cite{Leray}, de Gosson \cite{Birk}; for a detailed discussion of
the notion of generating function see Arnol'd \cite{Arnold}.
\end{proof}

The factorization $\widehat{S} = \widehat{S}_{W,m}\widehat{S}_{W^{\prime},
m^{\prime}}$ of a metaplectic operator is by no means unique; for instance we
can write the identity operator $I$ as $\widehat{S}_{W,m}\widehat{S}%
_{W,m}^{-1}$ $= \widehat{S}_{W,m}\widehat{S}_{W^{\ast},m^{\ast}}$ for every
quadratic Fourier transform $\widehat{S}_{W,m}$. There is however an invariant
attached to $\widehat{S}$: the Maslov index. Denoting by
$\operatorname*{Inert}R$ the index of inertia (= the number of negative
eigenvalues) of the real symmetric matrix $R$ we have:

\begin{proposition}
Let $\widehat{S}=\widehat{S}_{W,m}\widehat{S}_{W^{\prime},m^{\prime}%
}=\widehat{S}_{W^{\prime\prime},m^{\prime\prime}}\widehat{S}_{W^{^{\prime
\prime\prime}},m^{\prime\prime\prime}}$. We have%
\begin{equation}
m+m^{\prime}-\operatorname*{Inert}(P^{\prime}+Q)\equiv m^{\prime\prime
}+m^{\prime\prime\prime}-\operatorname*{Inert}(P^{\prime\prime\prime
}+Q^{\prime\prime})\text{ \ }\operatorname{mod}4. \tag{A8}\label{mm1}%
\end{equation}

\end{proposition}

\begin{proof}
See Leray \cite{Leray}, de Gosson \cite{AIF,Birk}.
\end{proof}

It follows from formula (\ref{mm1}) that the class modulo $4$ of the integer
$m+m^{\prime}-\operatorname*{Inert}(P^{\prime}+Q)$ does not depend on the way
we write $\widehat{S}\in\operatorname*{Mp}(n)$ as a product $\widehat{S}%
_{W,m}\widehat{S}_{W^{\prime},m^{\prime}}$ of quadratic Fourier transforms;
this class is denoted by $m(\widehat{S})$ and called the Maslov index of
$\widehat{S}$. The mapping
\[
m:\operatorname*{Mp}(n)\in\widehat{S}\longrightarrow m(\widehat{S}%
)\in\mathbb{Z}_{4}\text{ \ }%
\]
is called the Maslov index on $\operatorname*{Mp}(n)$. We have $m(\widehat{S}%
_{W,m})=m$, $\operatorname{mod}4$ (\cite{Leray,AIF}). The theory of the Maslov
index has been further developed by Arnol'd, Leray, and by the author (see the
review \cite{CLM} by Cappell \textit{et al.}).

\section{APPENDIX\ B: Leray and Maslov Indices}

\subsection*{B.1 The Leray index}

Let $\operatorname*{Lag}(n)$ be the Lagrangian Grassmannian of the symplectic
space $(\mathbb{R}^{2n},\sigma)$. We have a natural action%
\[
\operatorname*{Sp}(n)\times\operatorname*{Lag}(n)\longrightarrow
\operatorname*{Lag}(n)\text{.}%
\]
Let $(\ell,\ell^{\prime},\ell^{\prime\prime})\in\operatorname*{Lag}^{3}(n)$;
we denote by $\tau(\ell,\ell^{\prime},\ell^{\prime\prime})$ the signature of
the quadratic form%
\[
Q(z,z^{\prime},z^{\prime\prime})=\sigma(z,z^{\prime})+\sigma(z^{\prime
},z^{\prime\prime})+\sigma(z^{\prime\prime},z)
\]
on $\ell\times\ell^{\prime}\times\ell^{\prime\prime}$. It has the following properties:

\begin{itemize}
\item \textit{Symplectic invariance}:%
\[
\tau(S\ell,S\ell^{\prime},S\ell^{\prime\prime})=\tau(\ell,\ell^{\prime}%
,\ell^{\prime\prime})\text{ \ \textit{for all} }S\in\operatorname*{Sp}(n);
\]

\item \textit{Cocycle property}:%
\begin{equation}
\partial\tau(\ell,\ell^{\prime},\ell^{\prime\prime},\ell^{\prime\prime\prime
})=0 \tag{B1}\label{cocycletau}%
\end{equation}
where $\partial$ is the usual coboundary operator;

\item \textit{Antisymmetry}:%
\[
\tau(\pi(\ell,\ell^{\prime},\ell^{\prime\prime}))=(-1)^{\mathrm{sign}(\pi
)}\tau(\ell,\ell^{\prime},\ell^{\prime\prime})
\]
for every permutation $\pi$ of $(\ell,\ell^{\prime},\ell^{\prime\prime})$.
\end{itemize}

We have%
\begin{equation}
\tau(\ell,\ell^{\prime},\ell^{\prime\prime})\equiv n+\partial\dim(\ell
,\ell^{\prime},\ell^{\prime\prime})\text{ \ }\operatorname{mod}2
\tag{B2}\label{deltadim1}%
\end{equation}
where $\dim(\ell,\ell^{\prime})=\dim(\ell\cap\ell^{\prime})$.

Let $\pi_{\infty}:\operatorname*{Lag}_{\infty}(n)\longrightarrow
\operatorname*{Lag}(n)$ be the universal covering space of
$\operatorname*{Lag}(n)$ (\textquotedblleft Maslov bundle\textquotedblright).
We will write $\ell=\pi_{\infty}^{\operatorname*{Lag}}(\ell_{\infty})$. The
Leray index is the only mapping
\[
\mu:\operatorname*{Lag}\nolimits_{\infty}(n)\times\operatorname*{Lag}%
\nolimits_{\infty}(n)\longrightarrow\mathbb{Z}%
\]
having the two following properties:

\begin{description}
\item[LM1] \textit{It is locally constant on the set}
\[
\{(\ell_{\infty},\ell_{\infty}^{\prime})\in\operatorname*{Lag}%
\nolimits_{\infty}(n)\times\operatorname*{Lag}\nolimits_{\infty}(n):\ell
\cap\ell^{\prime}=0\};
\]

\item[LM2] \textit{Its coboundary descends to the signature}: $\partial
\mu(\ell_{\infty},\ell_{\infty}^{\prime},\ell_{\infty}^{\prime\prime}) =
\tau(\ell,\ell^{\prime},\ell^{\prime\prime})$, \textit{that is}
\begin{equation}
\mu(\ell_{\infty},\ell_{\infty}^{\prime})-\mu(\ell_{\infty},\ell_{\infty
}^{\prime\prime})+\mu(\ell_{\infty}^{\prime},\ell_{\infty}^{\prime\prime
})=\tau(\ell,\ell^{\prime},\ell^{\prime\prime}). \tag{B3}\label{cocycle2}%
\end{equation}

\end{description}

Taking $\ell_{\infty}=\ell_{\infty}^{\prime\prime}$ in (\ref{cocycle2}) and
using the antisymmetry of $\tau$ we get the relation%
\begin{equation}
\mu(\ell_{\infty},\ell_{\infty}^{\prime})=-\mu(\ell_{\infty}^{\prime}%
,\ell_{\infty}). \tag{B4}\label{invLeray}%
\end{equation}

Identifying as usual the unitary group $U(n,\mathbb{C})$ with a subgroup
$U(n)$ of $\operatorname*{Sp}(n)$ we have a transitive action
\[
U(n,\mathbb{C})\times\operatorname*{Lag}(n)\longrightarrow\operatorname*{Lag}%
(n).
\]
Let\ $\ell_{P}=0\times\mathbb{R}^{n}$ the mapping $\ell=u\ell_{P}\longmapsto
uu^{T}$ ($u\in U(n,\mathbb{C})$) induces a homeomorphism
\[
\operatorname*{Lag}(n)\longrightarrow W(n,\mathbb{C})=\{w\in U(n,\mathbb{C}%
):w=w^{T}\}
\]
and we have the identification with the set
\[
\operatorname*{Lag}\nolimits_{\infty}(n)\equiv\{(w,\theta):w\in W(n,\mathbb{C}%
),\,\det w=e^{i\theta}\};
\]
the projection $\pi_{\infty}^{\operatorname*{Lag}}$ is the mapping
$(w,\theta)\longmapsto w$. The Leray index can then be explicitly be defined
in the transversal case $\ell\cap\ell^{\prime}=0$ by the Souriau
\cite{Souriau} formula
\begin{equation}
\mu(\ell_{\infty},\ell_{\infty}^{\prime})=\frac{1}{\pi}(\theta-\theta^{\prime
}+i\operatorname*{Tr}\operatorname*{Log}(-w(w^{\prime})^{-1}) \tag{B5}%
\label{Souriau}%
\end{equation}
when $\ell_{\infty}=(w,\theta)$ and $\ell_{\infty}^{\prime}=(w^{\prime}%
,\theta^{\prime})$. The condition $\ell\cap\ell^{\prime}=0$ is equivalent to
$-w(w^{\prime})^{-1}$ having no eigenvalue on the negative half-axis. In the
non-transversal case one chooses $\ell_{\infty}^{\prime\prime}\in
\operatorname*{Lag}_{\infty}(n)$ such that $\ell^{\prime\prime}\cap\ell
=\ell^{\prime\prime}\cap\ell^{\prime}=0$ and one then defines%
\begin{equation}
\mu(\ell_{\infty},\ell_{\infty}^{\prime})=\mu(\ell_{\infty},\ell_{\infty
}^{\prime\prime})-\mu(\ell_{\infty}^{\prime},\ell_{\infty}^{\prime\prime
})+\tau(\ell,\ell^{\prime},\ell^{\prime\prime}). \tag{B6}\label{LeGo}%
\end{equation}
That the right-hand side in this formula is independent on the choice of
$\ell_{\infty}^{\prime\prime}$ readily follows from the cocycle property
(\ref{cocycletau}) of the signature $\tau$ \cite{JMPA,Birk}.

We have%
\begin{equation}
\mu(\ell_{\infty},\ell_{\infty}^{\prime})\equiv n+\dim(\ell\cap\ell^{\prime
})\text{ }\operatorname{mod}2\text{\ , \ }\mu(\ell_{\infty},\ell_{\infty
}^{\prime})=-\mu(\ell_{\infty}^{\prime},\ell_{\infty}) \tag{B7}\label{mll'}%
\end{equation}
(the first equality immediately follows from (\ref{cocycle2}) using
(\ref{deltadim1}) and the second by taking $\ell_{\infty}^{\prime\prime}%
=\ell_{\infty}$ in (\ref{cocycle2})). Let $\operatorname*{Sp}\nolimits_{\infty
}(n)$ be the universal covering group of $\operatorname*{Sp}(n)$. The natural
group action
\[
\operatorname*{Sp}\nolimits_{\infty}(n)\times\operatorname*{Lag}%
\nolimits_{\infty}(n)\longrightarrow\operatorname*{Lag}\nolimits_{\infty}(n)
\]
such that
\begin{equation}
(\alpha S_{\infty})\ell_{\infty}=S_{\infty}(\beta^{2}\ell_{\infty})=\beta
^{2}(S_{\infty}\ell_{\infty}) \tag{B8}\label{sabl}%
\end{equation}
where $\alpha$ and $\beta$ the generators of the cyclic groups $\pi
_{1}[\operatorname*{Sp}(n)]$ and $\pi_{1}[\operatorname*{Lag}(n)]$,
respectively \cite{Leray}. We have%
\begin{equation}
\mu(\beta^{r}\ell_{\infty},\beta^{r^{\prime}}\ell_{\infty}^{\prime})=\mu
(\ell_{\infty},\ell_{\infty}^{\prime})+2(r-r^{\prime}) \tag{B9}\label{murr}%
\end{equation}
for all $(r,r^{\prime})\in\mathbb{Z}^{2}$.

The Leray index is invariant under the action of $\operatorname*{Sp}%
\nolimits_{\infty}(n)$:%
\begin{equation}
\mu(S_{\infty}\ell_{\infty},S_{\infty}\ell_{\infty}^{\prime})=\mu(\ell
_{\infty},\ell_{\infty}^{\prime}). \tag{B10}\label{spaction}%
\end{equation}
This immediately follows from the fact that both functions $(\ell_{\infty
},\ell_{\infty}^{\prime})\longmapsto\mu(\ell_{\infty},\ell_{\infty}^{\prime})$
and $(\ell_{\infty},\ell_{\infty}^{\prime})\longmapsto\mu(S_{\infty}%
\ell_{\infty},S_{\infty}\ell_{\infty}^{\prime})$ satisfy the characteristic
conditions (LM1) and (LM2) and that the signature is a symplectic invariant.

\subsection*{B.2 Relative Maslov indices}

For $S_{\infty}\in\operatorname*{Sp}\nolimits_{\infty}(n)$ and $\ell
\in\operatorname*{Lag}(n)$ we define the Maslov index on $\operatorname*{Sp}%
\nolimits_{\infty}(n)$ relative to $\ell$ by
\begin{equation}
\mu_{\ell}(S_{\infty})=\mu(S_{\infty}\ell_{\infty},\ell_{\infty}).
\tag{B11}\label{mule}%
\end{equation}
It follows from (\ref{sabl}) that for every $\ell_{\infty}\in
\operatorname*{Lag}\nolimits_{\infty}(n)$ the function $\operatorname*{Sp}%
\nolimits_{\infty}(n)\longrightarrow\mathbb{Z}$ associating to $S_{\infty}$
the integer $\mu(S_{\infty}\ell_{\infty},\ell_{\infty})$ only depends on the
projection $\ell=\pi_{\infty}(\ell_{\infty})$, justifying the notation
(\ref{mule}).

Let $S_{\infty},S_{\infty}^{\prime}\in\operatorname*{Sp}\nolimits_{\infty}(n)$
and $\ell\in\operatorname*{Lag}(n)$. We have the product formula%
\begin{equation}
\mu_{\ell}(S_{\infty}S_{\infty}^{\prime})=\mu_{\ell}(S_{\infty})+\mu_{\ell
}(S_{\infty}^{\prime})+\tau(\ell,S\ell,SS^{\prime}\ell) \tag{B12}%
\label{product}%
\end{equation}
(it readily follows from the coboundary property (\ref{cocycle2}) of the Leray
index); taking $S_{\infty}^{\prime}=S_{\infty}^{-1}$ in this formula it
follows that%
\begin{equation}
\mu_{\ell}(S_{\infty}^{-1})=-\mu_{\ell}(S_{\infty}). \tag{B13}%
\label{inversion}%
\end{equation}

The following identity describes the action of $\pi_{1}[\operatorname*{Sp}%
(n)]$ on the relative Maslov index: for every $r\in\mathbb{Z}$ we have from
(\ref{sabl}) that
\begin{equation}
\mu_{\ell}(\alpha^{r}S_{\infty})=\mu_{\ell}(S_{\infty})+4r. \tag{B14}%
\label{alphamaslov}%
\end{equation}

It follows from the properties (LM1) and (LM2) of the Leray index that

\begin{description}
\item[MA] \textit{The Maslov index relative to} $\ell\in\operatorname*{Lag}%
(n)$ \textit{is the only mapping} $\mu_{\ell}:\operatorname*{Sp}%
\nolimits_{\infty}(n)$ $\longrightarrow\mathbb{Z}$ \textit{which is locally
constant on the set} $\{S_{\infty}:S\ell\cap\ell=0\}$ \textit{and satisfying
the product formula} (\ref{product}).
\end{description}

\begin{acknowledgement}
Maurice de Gosson has been financed by the Austrian FWF research grant
P27773--N25. Fernando Nicacio wishes to acknowledge financial support from the
Brazilian founding agency CAPES.
\end{acknowledgement}

\end{document}